\newif\ifarxiv

\arxivtrue

\ifarxiv
\documentclass{article}
\usepackage[numbers]{natbib}
\usepackage{amsfonts,amssymb,xcolor,booktabs,caption}
\usepackage[margin=1.25in]{geometry}
\else
\documentclass[manuscript]{acmart}
\fi

\AtBeginDocument{%
  }

\usepackage[algo2e,linesnumbered,noresetcount,noend]{algorithm2e}
\let\oldnl\nl
\newcommand{\nonl}{\renewcommand{\nl}{\let\nl\oldnl}}

\usepackage{amsmath,amsthm,enumitem,hyperref}
\usepackage[textsize=tiny,disable=true]{todonotes}

\newtheorem{theorem}{Theorem}
\newtheorem{lemma}[theorem]{Lemma}
\newtheorem{observation}[theorem]{Observation}
\newtheorem{claim}[theorem]{Claim}
\newtheorem{corollary}[theorem]{Corollary}

\newcommand{\widesim}[2][1.5]{
  \mathrel{\overset{#2}{\scalebox{#1}[1]{$\sim$}}}
}

\usetikzlibrary{
  positioning,
  automata,
  arrows,
  calc,
  decorations.pathmorphing,
  decorations.markings,
  decorations.pathreplacing,
  patterns,
  shapes.arrows,
  shapes.callouts,
  topaths
}

\tikzset{
    linpt/.style={fill,rectangle,yscale=0.085cm,xscale=0.012cm},
    >=stealth',
    pil/.style={
           ->,
           thick,
           shorten >=4pt
           }
}

\ifarxiv

\hypersetup{
pdftitle={A template for the arxiv style},
pdfsubject={q-bio.NC, q-bio.QM},
pdfauthor={David S.~Hippocampus, Elias D.~Striatum},
pdfkeywords={First keyword, Second keyword, More},
}
\fi

\begin{document}

\title{The Space Complexity of Consensus from Swap}

\ifarxiv
\author{Sean Ovens \\
	University of Toronto\\
	\texttt{sgovens@cs.toronto.edu} \\
}
\fi


\ifarxiv
\maketitle
\else
\fi


\begin{abstract}
Nearly thirty years ago, it was shown that $\Omega(\sqrt{n})$ read/write registers are needed to solve randomized wait-free consensus among $n$ processes.
This lower bound was improved to $n$ registers in 2018, which exactly matches known algorithms.
The $\Omega(\sqrt{n})$ space complexity lower bound actually applies to a class of objects called historyless objects, which includes registers, test-and-set objects, and readable swap objects.
However, every known $n$-process obstruction-free consensus algorithm from historyless objects uses $\Omega (n)$ objects.

In this paper, we give the first $\Omega (n)$ space complexity lower bounds on consensus algorithms for two kinds of historyless objects.
First, we show that any obstruction-free consensus algorithm from swap objects uses at least $n-1$ objects.
More generally, we prove that any obstruction-free $k$-set agreement algorithm from swap objects uses at least $\lceil \frac{n}{k}\rceil - 1$ objects.
The $k$-set agreement problem is a generalization of consensus in which processes agree on no more than $k$ different output values.
This is the first non-constant lower bound on the space complexity of solving $k$-set agreement with swap objects when $k > 1$.
We also present an obstruction-free $k$-set agreement algorithm from $n-k$ swap objects, which exactly matches our lower bound when $k=1$.

Second, we show that any obstruction-free binary consensus algorithm from readable swap objects with domain size $b$ uses at least $\frac{n-2}{3b+1}$ objects.
When $b$ is a constant, this asymptotically matches the best known obstruction-free consensus algorithms from readable swap objects with unbounded domains.
Since any historyless object can be simulated by a readable swap object with the same domain, our results imply that any obstruction-free consensus algorithm from historyless objects with domain size $b$ uses at least $\frac{n-2}{3b+1}$ objects.
For $b = 2$, we show a slightly better lower bound of $n-2$.
There is an obstruction-free binary consensus algorithm using $2n-1$ readable swap objects with domain size $2$, asymptotically matching our lower bound.
\end{abstract}


%
%
%


\section{Introduction}\label{sec:introduction}
Consensus is one of the most well-studied problems in distributed computing.
In the consensus problem, $n$ processes each begin with an input and they collectively try to agree on a single output that is equal to the input of some process.
A consensus algorithm is wait-free if every process decides within a finite number of its own steps.
Unfortunately, a well-known result by Fischer, Lynch, and Paterson \citep{flp-85} proves that deterministic wait-free consensus is unsolvable for $n \geq 2$ processes in the asynchronous message passing model, even when only one process can crash.
Wait-free consensus is also impossible in deterministic asynchronous shared memory when processes communicate using only read/write registers \citep{la-87,cil-87}.
Randomized wait-freedom is a weaker progress condition that requires each process to decide within a finite number of its own steps in expectation.
There are known randomized wait-free consensus algorithms from $n$ registers \citep{ah-90,cil-94}.

In 1993, Ellen, Herlihy, and Shavit \citep{ehs-98} proved that $\Omega (\sqrt{n})$ registers are needed to solve randomized wait-free consensus.
This lower bound applies to binary consensus algorithms, where process inputs are all either $0$ or $1$, so it applies to consensus algorithms with arbitrary inputs as well.
Over twenty years later, a breakthrough result by Zhu \citep{z-16} showed that any randomized wait-free binary consensus algorithm from registers uses at least $n-1$ registers.
Finally, using a novel technique, Ellen, Gelashvili, and Zhu \citep{egz-18} improved this lower bound to $n$ registers, exactly matching the space complexity of known algorithms \citep{ah-90,cil-94}.

All of these lower bounds apply to \emph{nondeterministic solo-terminating} consensus algorithms.
Such algorithms have the property that, for every reachable configuration $C$ of the algorithm and every process $p$, there is a solo execution by $p$ from $C$ in which $p$ decides.
Every randomized wait-free algorithm is nondeterministic solo-terminating.
If a nondeterministic solo-terminating algorithm is deterministic, then it is called \emph{obstruction-free}.
Hence, space lower bounds that are proved for nondeterministic solo-terminating consensus algorithms also apply to randomized wait-free and obstruction-free consensus algorithms.

Ellen, Gelashvili, and Zhu \citep{egz-18} proved that any nondeterministic solo-terminating algorithm from objects that support read can be transformed into an obstruction-free algorithm using the same objects.
Hence, space lower bounds for obstruction-free algorithms using objects that support read also apply to nondeterministic solo-terminating (and hence, randomized wait-free) algorithms.
Obstruction-freedom is a simpler property than nondeterministic solo-termination and randomized wait-freedom, so this is helpful for proving new space complexity lower bounds.

Herlihy's wait-free consensus hierarchy \citep{h-91} is one way to compare the relative power of various objects.
An object has consensus number $x$ if there is an $x$-process, deterministic, wait-free consensus algorithm from instances of that object and registers, but there is no such algorithm for more than $x$ processes.
This hierarchy collapses in a randomized (or obstruction-free) setting, since there are known $n$-process randomized wait-free (and obstruction-free) consensus algorithms from registers \citep{ah-90,cil-94}.
Ellen, Gelashvili, Shavit, and Zhu \citep{egsz-20} proposed an alternative classification of objects based on the number of instances of the object required to solve obstruction-free consensus for $n$ processes.
Since space lower bounds for obstruction-free implementations from readable objects also apply to randomized wait-free implementations \citep{egz-18}, this also gives us a way to compare the relative power of objects in a randomized setting.

\smallskip

Registers belong to a class of objects called historyless objects.
A historyless object has the property that its value depends only on the last nontrivial operation applied to it.
An operation is trivial if it cannot modify the value of the object.
Registers support the trivial operation \emph{Read}, which returns the current value of the object, and the nontrivial operation \emph{Write}$(v)$, which always sets the value of the object to $v$.
Another example of a historyless object is a swap object.
A swap object supports the \emph{Swap}$(v)$ operation, which atomically changes the value of the object to $v$ and returns its previous value.
Readable swap objects support \emph{Swap}$(v)$ and \emph{Read}.
Any historyless object can be simulated by a readable swap object \citep{efr-07}.
Hence, when proving space complexity lower bounds for algorithms that use historyless objects, it suffices to consider readable swap objects.
It is impossible to solve wait-free consensus among $n \geq 3$ processes using only historyless objects \citep{h-91}.
Ellen, Gelashvili, Shavit, and Zhu \citep{egsz-20} gave an obstruction-free consensus algorithm using $n-1$ readable swap objects.
Their algorithm is similar to Aspnes and Herlihy's racing counters algorithm from $n$ registers \citep{ah-90}.


The $\Omega (\sqrt{n})$ lower bound by Ellen, Herlihy, and Shavit \citep{ehs-98} actually applies to nondeterministic solo-terminating consensus algorithms that use only historyless objects, even when the objects can have unbounded domain size.
In fact, this is still the best space complexity lower bound that is known for consensus algorithms that use only historyless objects.


In this paper, we approach this longstanding gap by considering algorithms that use readable swap objects with bounded domain sizes.
We prove that any $n$-process nondeterministic solo-terminating binary consensus algorithm using readable swap objects with domain size $b$ uses at least $\frac{n-2}{3b+1}$ objects.
When $b$ is a constant, our lower bound differs from Ellen, Gelashvili, Shavit, and Zhu's upper bound by only a constant factor.
When $b = 2$, we give a slightly better lower bound of $n-2$ readable swap objects.
There is an obstruction-free binary consensus algorithm from $2n-1$ registers with domain size $2$ \citep{b-11}, asymptotically matching our lower bound.

The proofs of our lower bounds use a new technique that we first used to prove space complexity lower bounds for scannable objects \citep{me-21}.
A scannable object is a generalization of a snapshot object \citep{aadgms-93,a-93}.
It consists of a sequence of readable objects called components that can all be read simultaneously.
A scannable object is an example of a long-lived object, which means that each process can apply arbitrarily many operations to the object.
In contrast, in a binary consensus algorithm, each process has a single input and, if it does not crash, it produces a single output and terminates (i.e. does not continue to participate in the execution).
For this reason, we could not use our lower bound technique for scannable objects to directly obtain our result for obstruction-free consensus.
In this paper, we use a novel approach that combines this technique with a covering argument and a valency argument, which are both standard techniques for proving lower bounds.
Our results, along with some known upper and lower bounds on the space complexity of solving consensus using historyless objects, are summarized in Table~\ref{tab:setagree}.

\medskip

In the $k$-set agreement problem \cite{c-93}, processes must collectively agree on at most $k$ distinct output values.
Consensus is the same as $1$-set agreement.
When $n \leq k$, each process can simply output its input value.
On the other hand, when $n > k$ it is known that deterministic wait-free $k$-set agreement is unsolvable in the asynchronous shared memory model when processes communicate using only registers \cite{bg-93,hs-99,sz-00}.
There is a simple obstruction-free $k$-set agreement algorithm using $n-k+1$ registers: $n-k+1$ processes use the registers to solve consensus, and the remaining $k-1$ processes decide their input values.
Bouzid, Raynal, and Sutra \citep{brs-18} obtained the same result when processes are anonymous, meaning that processes run the same protocol, do not have identifiers, and their initial states depend only on their inputs.
Ellen, Gelashvili, and Zhu \citep{egz-18} proved that at least $\lceil\frac{n}{k}\rceil$ registers are required to solve obstruction-free $k$-set agreement using only registers when $n > k \geq 1$.
They conjectured that $n-k+1$ registers are required.

Delporte-Gallet, Fauconnier, Kuznetsov, and Ruppert \citep{dfkr-15} proved a lower bound of $n-k+1$ registers for solving repeated obstruction-free $k$-set agreement.
This problem is an unbounded sequence of independent instances of $k$-set agreement.
They also showed that repeated obstruction-free $k$-set agreement can be solved with $\textit{min}(n-k+2, n)$ registers and anonymously with $2(n-k) + 2$ registers.
Bouzid, Raynal, and Sutra \citep{brs-18} later showed that repeated obstruction-free $k$-set agreement can be solved anonymously with $n-k+1$ registers, exactly matching the lower bound.

There is also a simple wait-free $2$-process consensus algorithm from a single swap object.
The swap object initially contains a special value $\bot$ which cannot be the input value of any process.
Both processes swap their input value into the object.
The process that receives the response $\bot$ decides its input value and the other process decides the value it obtained in response to its swap operation.
Using this $2$-process consensus algorithm and a reduction by Chaudhuri and Reiners \citep{cr-96}, we can construct a simple wait-free $n$-process $k$-set agreement algorithm from $n-k$ swap objects when $k \geq \lceil\frac{n}{2}\rceil$ as follows: $n-k$ different pairs of processes each use a different swap object to solve consensus, while the remaining $2k - n$ processes simply decide their input values.
It is unknown whether wait-free $k$-set agreement is solvable using swap objects (or readable swap objects) when $\lceil\frac{n}{2}\rceil > k > 1$.
Furthermore, prior to this paper, there was no known non-constant lower bound on the space complexity of solving nondeterministic solo-terminating $k$-set agreement using swap objects when $n > k > 1$.



We use an indistinguishability argument to show that any $n$-process nondeterministic solo-terminating $k$-set agreement algorithm using swap objects requires at least $\lceil\frac{n}{k}\rceil - 1$ objects.
Our proof of this lower bound is concise, and we believe it offers insight into why nontrivial historyless operations alone have limited power to solve set agreement.
We also give an $n$-process obstruction-free $k$-set agreement algorithm from $n-k$ swap objects, exactly matching our lower bound for $k = 1$.
When $k < \lceil\frac{n}{2}\rceil$, this is the first known obstruction-free $k$-set agreement algorithm from swap objects.
Our algorithm is based on the $n$-process obstruction-free consensus algorithm from $n-1$ readable swap objects by Ellen, Gelashvili, Shavit, and Zhu \citep{egsz-20}, but the proofs of correctness are completely different.
Our results, along with some known upper and lower bounds on the space complexity of solving $k$-set agreement using historyless objects, are summarized in Table~\ref{tab:setagree}.

\renewcommand{\arraystretch}{1.5}
\begin{table}[h]
\centering
\begin{tabular}{lp{3.9cm}cc}\toprule
Task & Objects used by the algorithm		&	Lower bound			&	Upper bound \\ \hline
Consensus & Registers	&	$n$ \citep{egz-18}	&	$n$ \citep{ah-90,cil-94} \\
Consensus & Swap objects	&	$\mathbf{n-1}$ \textbf{[Theorem~\ref{thm:swaplb}]}	&	$\mathbf{n-1}$ \textbf{[Algorithm~\ref{alg:setagreeswap}]} \\
Consensus & Readable swap objects with domain size $2$	&	$\mathbf{n-2}$ \textbf{[Theorem~\ref{thm:bswaplb}]}	&	$2n-1$ \citep{b-11} \\
Consensus & Readable swap objects with domain size $b$	&	$\mathbf{\frac{n-2}{3b+1}}$ \textbf{[Theorem~\ref{thm:boundedlb}]}	&	$2n-1$ \citep{b-11} \\
Consensus & Readable swap objects with unbounded domain	&	$\Omega(\sqrt{n})$ \citep{ehs-98}	&	$n-1$ \citep{egsz-20} \\
$k$-set agreement & Registers	&	$\lceil\frac{n}{k}\rceil$ \citep{egz-18}	&	$n-k+1$ \citep{brs-18} \\
$k$-set agreement & Swap objects	&	$\mathbf{\lceil\frac{n}{k}\rceil-1}$ \textbf{[Theorem~\ref{thm:swaplb}]}	&	$\mathbf{n-k}$ \textbf{[Algorithm~\ref{alg:setagreeswap}]} \\
$k$-set agreement & Readable swap objects with unbounded domain	&	$1$	&	$\mathbf{n-k}$ \textbf{[Algorithm~\ref{alg:setagreeswap}]} \\
\bottomrule
\end{tabular}

\vspace{2mm}

\caption{\textmd{Lower and upper bounds on the space complexity of solving $n$-process nondeterministic solo-terminating consensus when $n > 1$ and $k$-set agreement when $n > k > 1$ with different kinds of historyless objects. Our new results are in \textbf{boldface}.}}\label{tab:setagree}
\end{table}


\medskip

We present our model of computation together with a brief discussion of covering arguments and valency arguments in Section~\ref{sec:model}.
In Section~\ref{sec:unreadableub}, we present our obstruction-free $k$-set agreement algorithm from swap objects.
In Section~\ref{sec:unreadablelb}, we prove our lower bound on the number of swap objects needed to solve nondeterministic solo-terminating $k$-set agreement.
In Section~\ref{sec:readable}, we present our lower bound on the space complexity of obstruction-free consensus algorithms from readable swap objects with bounded domain sizes.
Finally, we conclude and discuss some possible research directions in Section~\ref{sec:conclusion}.

\section{Preliminaries}\label{sec:model}

We consider a standard asynchronous shared memory model in which $n$ processes communicate using instances of shared \emph{objects} provided by the system.
An object has a set of possible \emph{values}, a set of \emph{operations} that can be applied to it, and a set of \emph{responses} that these operations can return.

A \emph{swap object} stores a value $v \in \mathbb{N}$ and supports the \emph{Swap}$(v')$ operation, which returns the current value $v$ of the object and changes its value to $v'$.
A \emph{readable binary swap object} stores a value $v \in \{0, 1\}$ and supports \emph{Read} (which returns the current value of the object), \emph{Swap}$(0)$, and \emph{Swap}$(1)$.
If $v \in \{0, 1\}$, then we use $\bar{v}$ to denote the value $1 - v$.

In the \emph{$k$-set agreement} problem, each process is given some input value, and processes attempt to collectively agree on no more than $k$ output values.
When a process $p$ outputs $v$, we say that $p$ \emph{decides} the value $v$.
A \emph{$k$-set agreement algorithm} consists of a set of objects and a procedure for each process, and must satisfy the following two properties.
\begin{itemize}
	\item \emph{$k$-Agreement}: no more than $k$ values are decided.
	\item \emph{Validity}: if a process decides the value $v$, then $v$ was the input value of some process.
\end{itemize}
In the \emph{$m$-valued $k$-set agreement} problem, process inputs are from the set $\{0, \ldots, m-1\}$.
Notice that the $m$-valued $k$-set agreement problem is trivial if $m \leq k$.

The \emph{consensus} problem is another name for the $1$-set agreement problem.
We will simply use \emph{agreement} to refer to the $1$-agreement property of consensus algorithms.
The $2$-valued consensus problem is called the \emph{binary consensus} problem.

A \emph{configuration} of a $k$-set agreement algorithm consists of a state for every process and a value for every object.
We use $\textit{value}(B, C)$ to denote the value of the object $B$ in the configuration $C$.
A \emph{step} by a process consists of an operation applied to some object, a response to that operation, and some finite amount of local computation by that process.

An \emph{execution} is an alternating sequence of configurations and steps, beginning with a configuration, such that each step is applied in the configuration that precedes it and results in the configuration that follows it.
A finite execution ends with a configuration.
If $C$ is a configuration and $\alpha$ is a finite execution from $C$, then $C\alpha$ denotes the final configuration in $\alpha$.
An execution $\alpha$ is $P$-only, where $P$ is some set of processes, if every step in $\alpha$ is applied by a process in $P$.
If $P = \{p_j\}$, then we say $\alpha$ is $p_j$-only.
A \emph{solo-terminating} execution by process $p_j$ from a configuration $C$ is a $p_j$-only execution that ends with a configuration in which $p_j$ has decided a value.

For every configuration $C$ and every process $p$, a $k$-set agreement algorithm specifies the next operation that $p$ will apply given its state in $C$.
We say that $p$ is \emph{poised} to apply this operation in $C$.
An execution is produced by a \emph{scheduler}, which decides the order in which processes take steps.
That is, for any configuration $C$ of a $k$-set agreement algorithm, a scheduler picks a process $p$ that has not decided in $C$ to take its next step.
Suppose that process $p$ is poised to apply the operation $op$ to the object $B$ in configuration $C$.
If the scheduler picks process $p$ to take a step in configuration $C$, then $p$ applies the operation $op$ to $B$ and obtains a response to $op$ based on the value of $B$ in $C$.
After applying $op$, process $p$ does some local computation, and then updates its own state based on the response to its operation and its local computation.
This results in a new configuration.
If the algorithm is \emph{deterministic}, then $p$ has exactly one possible state resulting from its local computation.
If the algorithm is \emph{randomized}, then the local computation could include some coin flips, so $p$ could have multiple possible states resulting from its local computation.
An \emph{initial configuration} defines the values of the objects before processes have taken any steps.

Two configurations $C_1$ and $C_2$ are \emph{indistinguishable} to a set of processes $P$ if every process in $P$ has the same state in $C_1$ and $C_2$.
This is denoted by $C_1 \widesim{P} C_2$.
Let $C_1$ and $C_2$ be configurations of a $k$-set agreement algorithm.
Let $\alpha_1$ and $\alpha_2$ be executions starting from $C_1$ and $C_2$, respectively.
Then $\alpha_1$ and $\alpha_2$ are \emph{indistinguishable} to a set of processes $P$ if $C_1 \widesim{P} C_2$ and every process in $P$ performs the same sequence of steps in $\alpha_1$ and $\alpha_2$.
(Thus, every process in $P$ obtains the same sequence of responses to all of its operations and local coin flips in $\alpha_1$ and $\alpha_2$.)
This is denoted by $\alpha_1 \widesim{P} \alpha_2$.
If $C_1 \widesim{P} C_2$ and $\alpha_1 \widesim{P} \alpha_2$, then $C_1\alpha_1 \widesim{P} C_2\alpha_2$.

Suppose that $C$ and $C'$ are configurations of a $k$-set agreement algorithm such that $C \widesim{P} C'$, for some set of processes $P$.
If $\alpha$ is a $P$-only execution from $C$ and the objects accessed by $P$ during $\alpha$ have the same values in $C$ and $C'$, then there is a $P$-only execution $\alpha'$ from $C'$ such that $\alpha \widesim{P} \alpha'$ \citep{ae-14}.


The \emph{history} of an execution is its sequence of operations along with the processes that applied them.\todo{do I still use this term anywhere?}
A history $\sigma$ is \emph{applicable} to $C$ if there is an execution starting from $C$ whose history is $\sigma$.
If $\sigma$ is finite, $C$ is a configuration of a deterministic $k$-set agreement algorithm, and $\sigma$ is applicable to $C$, then there is exactly one execution starting from $C$ whose history is $\sigma$.
We use $C\sigma$ to denote the final configuration in this execution.
A history $\sigma$ is $P$-only if it only contains operations by processes in $P$.
If $P = \{p_j\}$, then we say $\sigma$ is $p_j$-only.

A $k$-set agreement algorithm is \emph{nondeterministic solo-terminating} if, for every configuration $C$ of the algorithm and every process $p$, there is a solo-terminating execution by $p$ from $C$.
A nondeterministic solo-terminating $k$-set agreement algorithm that is deterministic is called \emph{obstruction-free}.
A $k$-set agreement algorithm is \emph{randomized wait-free} if, for every scheduler, the expected length of an execution produced by that scheduler is finite, where the expectation is taken over the local coin flips performed by all processes.

\medskip

Covering arguments are a standard technique for obtaining space lower bounds in distributed computing.
The first covering argument was used by Burns and Lynch \citep{bl-93} to prove that any mutual exclusion algorithm for $n \geq 2$ processes requires at least $n$ registers.
A set of processes $\mathcal{Q}$ \emph{covers} a set of registers $\mathcal{B}$ if $|\mathcal{Q}| = |\mathcal{B}|$ and, for every $B \in \mathcal{B}$, there is a process in $\mathcal{Q}$ that is poised to write to $B$ in its next step.
A block write by $\mathcal{Q}$ is an execution that consists of the next step by each process in $\mathcal{Q}$, applied consecutively.
This sets the registers in $\mathcal{B}$ to fixed values.
Immediately before the block write by $\mathcal{Q}$, we can insert any execution $\alpha$ not involving $\mathcal{Q}$ that only accesses the registers in $\mathcal{B}$.
The block write by $\mathcal{Q}$ hides $\alpha$ from processes that did not take steps during $\alpha$.
At best, a covering argument obtains a configuration in which every process covers a distinct object, which gives a lower bound of $n$.

Covering arguments can be generalized to historyless objects, where a set of processes $\mathcal{Q}$ covers a set of objects $\mathcal{B}$ if $|\mathcal{Q}| = |\mathcal{B}|$ and, for each object $B \in \mathcal{B}$, there is a process in $\mathcal{Q}$ that is poised to apply a nontrivial operation to $B$ in its next step.
A block update by $\mathcal{Q}$ is a generalization of a block write, where each of the processes in $\mathcal{Q}$ takes its next step.
When the block update by $\mathcal{Q}$ is applied after $\alpha$, the processes in $\mathcal{Q}$ may obtain information about $\alpha$.
Hence, we cannot reuse the processes in $\mathcal{Q}$ while hiding $\alpha$ from the other processes.
This makes it more difficult to apply covering arguments in systems with readable swap objects compared to systems with registers.

Valency arguments were first introduced by Fischer, Lynch, and Paterson \citep{flp-85}.
The valency of a set of processes in a configuration of a binary consensus algorithm is the set of values that can be output by those processes in executions from that configuration.
More formally, a set of processes $\mathcal{P}$ is \emph{bivalent} in configuration $C$ if, for each $v \in \{0, 1\}$, there exists an execution from $C$ only involving steps by $\mathcal{P}$ in which some process in $\mathcal{P}$ decides the value $v$.
If $\mathcal{P}$ is not bivalent in $C$, then it is \emph{univalent} in $C$.
More specifically, the set of processes $\mathcal{P}$ is $v$-\emph{univalent} in $C$ if, in every execution from $C$ that only includes steps by $\mathcal{P}$ in which some process in $\mathcal{P}$ decides, $v$ is the only value that is decided by any process in $\mathcal{P}$.


\section{Set Agreement Algorithm from Swap}\label{sec:unreadableub}

In this section, we present an obstruction-free $m$-valued $k$-set agreement algorithm from $n-k$ swap objects $B_1, \ldots, B_{n-k}$.
We emphasize that a swap object does not support the \emph{Read} operation.
In our algorithm, every swap object has two fields.
The \emph{lap counter} field consists of an array of $m$ values, all initially $0$, and the \emph{identifier} field consists of a single value, initially $\bot$.

We can view the algorithm as a race among the input values.
Every process $p$ stores a local lap counter $U[0,\ldots,m-1]$ that holds the highest lap for each input value that has been observed by $p$.
If $U[x] > U[x']$, for some $x, x' \in \{0, \ldots, m-1\}$, then $p$ believes that the value $x$ is ahead of the value $x'$ in the race.
During the algorithm, process $p$ repeatedly attempts to complete a lap for a value that appears to be leading the race.
To complete a lap, process $p$ must observe its own local lap counter and process identifier in every object.
When $p$ sees that some value $v$ is sufficiently far ahead of all the other values, it decides the value $v$.

\begin{algorithm2e}
\caption{An obstruction-free, $m$-valued, $k$-set agreement algorithm from $n-k$ swap objects.}
\label{alg:setagreeswap}

\SetAlgoNoLine

\SetKwProg{Fn}{Function}{:}{}
	\DontPrintSemicolon
	\SetKwFor{Loop}{loop}{}{end}
	
	\nonl\textbf{shared:} swap objects $B_1, \ldots, B_{n-k}$, initially $B_1 = \ldots = B_{n-k} = \bigl\langle [0, \ldots, 0], \bot \bigr\rangle$\;
	\Fn{$propose(v)$ by process $p$}{
		$U[0, \ldots, m-1] \gets [0, \ldots, 0]$\;\label{ln:init1}
		$U[v] \gets 1$\;\label{ln:init2}
		\Loop{}{\label{ln:mainloop}
			$\textit{conflict} \gets \textsc{False}$\;\label{ln:setconflict}
			\For{$i \in \{1, \ldots, n-k\}$}{\label{ln:swaploop}
				$\langle U', p'\rangle \gets Swap(B_i, \langle U, p\rangle)$\;\label{ln:swap}
				\If{$\langle U', p'\rangle \neq \langle U, p\rangle$}{\label{ln:idcond}
					$\textit{conflict} \gets \textsc{True}$\;
					\If{$U \neq U'$}{
						\For{all $j \in \{0, \ldots, m-1\}$} {\label{ln:updatelocal}
							$U[j] \gets$ max$\bigl(U[j], U'[j]\bigr)$\;\label{ln:setmaxlocal}
						}
					}
				}
			}\label{ln:endswaploop}
			\If{$\textit{conflict} = \textsc{False}$}{\label{ln:conflictcond}
				$c \gets$ max$(U)$\;\label{ln:maxval}
				$v \gets$ min$\bigl\{j\; :\; U[j] = c\bigr\}$\;\label{ln:minindex}
				\If{for all $j \neq v$, $U[v] \geq U[j] + 2$}{\label{ln:decidecond}
					\textbf{decide} $v$\;\label{ln:decide}
					\Return\;
				}
				\Else{
					$U[v] \gets U[v] + 1$\;\label{ln:inc}
				}
			}
		}\label{ln:endmainloop}
		
	}
\end{algorithm2e}

Algorithm~\ref{alg:setagreeswap} is a pseudocode description of our algorithm.
A process $p$ with input $v$ begins by initializing its local lap counter $U$ so that $U[j] = 0$ for all $j \in \{0, \ldots, m-1\}$, and then sets $U[v] = 1$.
Processes only decide and complete laps for values that they believe are winning the race, so this initialization step ensures that the algorithm satisfies validity.
After initializing its local lap counter $U$, process $p$ then repeatedly performs the loop on lines~\ref{ln:mainloop}--\ref{ln:endmainloop}.
An iteration of this loop begins with $p$ initializing a local Boolean variable \emph{conflict} to \textsc{False}.
This variable is used to indicate whether or not $p$ has observed a lap counter or a process identifier different from its own in some object.
During the loop on lines~\ref{ln:swaploop}--\ref{ln:endswaploop}, $p$ swaps its local lap counter and its identifier into the objects $B_1, \ldots, B_{n-k}$ one at a time.
If $p$ obtains a response different than $\langle U, p\rangle$ from one of these swaps, it sets its \emph{conflict} variable to \textsc{True}.
If it observes a lap counter $U' \neq U$, then process $p$ updates every component $j$ of its local lap counter to the maximum of $U[j]$ and $U'[j]$ on lines~\ref{ln:updatelocal}--\ref{ln:setmaxlocal}.
When $p$ reaches the end of the loop on lines~\ref{ln:mainloop}--\ref{ln:endmainloop} and \emph{conflict} is \textsc{True}, $p$ sets \emph{conflict} to \textsc{False} and restarts the loop on lines~\ref{ln:mainloop}--\ref{ln:endmainloop}.


When $p$ reaches the end of the loop on lines~\ref{ln:mainloop}--\ref{ln:endmainloop} and \emph{conflict} is \textsc{False}, $p$ must have observed its local lap counter and process identifier as the response to all $n-k$ swaps in the loop on lines~\ref{ln:swaploop}--\ref{ln:endswaploop}.
In this case, $p$ completes a lap.
Process $p$ chooses a value $v$ that appears to be leading the race, i.e. $U[v] \geq U[v']$ for all $v' \in \{0, \ldots, m-1\}$.
If there are multiple values that appear to be leading the race, then $p$ chooses $v$ to be the smallest of these values (lines~\ref{ln:maxval}--\ref{ln:minindex}).
Then, $p$ checks whether $v$ is at least $2$ laps ahead of all the other values on line~\ref{ln:decidecond}.
If so, then $p$ decides $v$ and returns.
Otherwise, $p$ increments the $v$-th component of its local lap counter on line~\ref{ln:inc}.

Consider any configuration $C$.
Let $M$ be the component-wise maximum of the lap counters in the swap objects and the local lap counter of process $p$ in $C$.
Notice that, if process $p$ runs on its own for sufficiently long starting from $C$, then its local lap counter will eventually contain the value $M$.
If $p$ continues to run on its own, then it will eventually swap $\langle M, p\rangle$ into every swap object.
Afterwards, if $p$ does another iteration of the loop on lines~\ref{ln:mainloop}--\ref{ln:endmainloop} on its own, it will complete a lap for the value $v$.
After completing at most $3$ laps for the value $v$, process $p$ will decide the value $v$.
Hence, the algorithm is obstruction-free.
We give a more rigourous proof of obstruction-freedom in Lemma~\ref{lem:of}.

\medskip

We will now formally prove the correctness of Algorithm~\ref{alg:setagreeswap}.
First, we make some simple observations about the algorithm.
If a process $p$ increases the value of the $j$-th component of its local lap counter to some value $\ell \geq 1$ on line~\ref{ln:setmaxlocal}, then the swap object accessed previously by $p$ must have contained a lap counter whose $j$-th component was $\ell$.
Some other process must have swapped this lap counter into the object in an earlier step.
That process must have had a local lap counter whose $j$-th component contained $\ell$ immediately before it applied this \emph{Swap} operation.
Therefore, the maximal value of the $j$-th component of the local lap counters across all processes cannot be changed when a process performs line~\ref{ln:setmaxlocal}.
More specifically, it can only be changed by a process performing line~\ref{ln:inc}.
This gives us the following.

\begin{observation}\label{obs:onebyone}
	Suppose that the $j$-th component of the local lap counter of some process has value $\ell \geq 2$ in a configuration $C$ of an execution $\alpha$.
	Then, for all $2 \leq \ell' \leq \ell$, there must be some step before $C$ in $\alpha$ by some process in which it increments the $j$-th component of its local lap counter on line~\ref{ln:inc} from $\ell'-1$ to $\ell'$.
\end{observation}

Process $p$ is the only process that writes its own identifier to any object.
Hence, if $p$ completes a lap, then all of the objects $B_1, \ldots, B_{n-k}$ must have contained the value of $p$'s local lap counter and its identifier immediately before $p$ began the loop on lines~\ref{ln:swaploop}--\ref{ln:endswaploop}.
We say a configuration $C$ is \emph{$\langle V, p\rangle$-total} if the value of every object $B_1, \ldots, B_{n-k}$ is $\langle V, p\rangle$ in $C$ and the value of $p$'s local lap counter is $V$ in $C$.
This gives us the following observation.

\begin{observation}\label{obs:total}
Suppose that a process $p$ takes a step in which it completes a lap.
Let $V$ be the value of $p$'s local lap counter immediately before this step.
Then the configuration immediately after $p$ performed line~\ref{ln:setconflict} for the last time before completing the lap was $\langle V, p\rangle$-total, and during the loop on lines~\ref{ln:swaploop}--\ref{ln:endswaploop}, $p$ swapped $\langle V, p\rangle$ into $B_1, \ldots, B_{n-k}$ and obtained $\langle V, p\rangle$ as the response to each of these operations.
\end{observation}

If $V$ and $V'$ are two lap counter values, then we say that $V$ is \emph{dominated by} $V'$, or $V \preceq V'$, if and only if $V[j] \leq V'[j]$, for all $j \in \{0, \ldots, m-1\}$.
After initializing its local lap counter $U$, a process $p$ only modifies $U$ on lines~\ref{ln:setmaxlocal} and \ref{ln:inc}.
These lines may only increase the values of the components of $p$'s local lap counter.
Hence, the value of $p$'s local lap counter before taking any step is dominated by the value of its local lap counter after taking the step.

\begin{observation}\label{obs:alwaysforward}
	If process $p$ has a local lap counter $V$ in a configuration $C$ of an execution $\alpha$, and in some later configuration $C'$ of $\alpha$, $p$ has local lap counter $V'$, then $V \preceq V'$.
\end{observation}

Following the initialization on lines~\ref{ln:init1}--\ref{ln:init2}, the values of the components of the local lap counters of each process are nonnegative.
Hence, the condition on line~\ref{ln:decidecond} implies that $p$'s preference must be on lap $2$ or greater when it decides.
This gives us the following.

\begin{observation}\label{obs:atleast2}
	If $V$ is the value of $p$'s local lap counter when it decides the value $x$, then $V[x] \geq 2$.
\end{observation}

In the next lemma, we show that if there is a $\langle V, p\rangle$-total configuration $C$ followed by a $\langle V', p'\rangle$-total configuration $C'$ with $V \not\preceq V'$ in some execution, then all of the objects must have been swapped by distinct processes between these two configurations.
Furthermore, when these processes apply their swap operations between $C$ and $C'$, they obtain lap counters that dominate $V$.
This implies that their local lap counters dominate $V$ in $C'$.
We formalize this in the following lemma.

\begin{lemma}\label{lem:manyprocesses}
%
%
%
%
Consider some $\langle V, p\rangle$-total configuration $C$ in an execution $\alpha$.
Let $C'$ be a $\langle V', p'\rangle$-total configuration that appears after $C$ in $\alpha$.
If $V \not\preceq V'$, then there are $n-k$ distinct processes different from $p$ and $p'$ that apply \emph{Swap} operations to different objects between $C$ and $C'$, and the values of the local lap counters of these processes in $C'$ dominate $V$.
\end{lemma}

\begin{proof}
Since $B_1, \ldots, B_{n-k}$ all contain $\langle V, p\rangle$ in $C$ and $\langle V', p'\rangle \neq \langle V, p\rangle$ in $C'$, every object must have had a lap counter that does not dominate $V$ swapped into it between $C$ and $C'$.
Let $i \in \{1, \ldots, n-k\}$ and consider the first process $q_i$ to perform a \emph{Swap}$(B_i, \langle V'', q_i\rangle)$ operation such that $V \not\preceq V''$ between $C$ and $C'$.
By Observation~\ref{obs:alwaysforward}, $q_i$ cannot be $p$, because $p$ has local lap counter $V$ in configuration $C$.
In the response to this operation, $q_i$ obtains a lap counter that dominates $V$.
Then $q$ updates its local lap counter to a value that dominates $V$ in the loop on lines~\ref{ln:updatelocal}--\ref{ln:setmaxlocal}.
Hence, every subsequent lap counter swapped into any object by $q_i$ dominates $V$.
Therefore, the processes $q_1, \ldots, q_{n-k}$ are distinct.
Furthermore, by Observation~\ref{obs:alwaysforward}, the local lap counters of $q_1, \ldots, q_{n-k}$ in $C'$ dominate $V$.
Since the local lap counter of $p'$ in $C'$ is $V'$ and $V \not\preceq V'$, this implies that $p' \not\in \{q_1, \ldots, q_{n-k}\}$.

\end{proof}

\begin{lemma}\label{lem:agree}
Algorithm~\ref{alg:setagreeswap} satisfies $k$-agreement.
\end{lemma}

\begin{proof}
To obtain a contradiction, suppose that there exists an execution of the algorithm in which at least $k+1$ distinct values are decided.
Let $\alpha$ be a prefix of such an execution in which exactly $k+1$ distinct values are decided.
For all $1 \leq j \leq k+1$, let $x_j$ be the $j$-th value decided during $\alpha$, and let $p_j$ be the first process to decide $x_j$ during $\alpha$.
Let $V_j$ be the value of $p_j$'s lap counter when it decides.
By Observation~\ref{obs:total}, the configuration $C_j$ immediately after $p_j$ performed line~\ref{ln:setconflict} for the last time is $\langle V_j, p_j\rangle$-total, and during its last execution of the loop on lines~\ref{ln:swaploop}--\ref{ln:endswaploop} $p_j$ swaps $\langle V_j, p_j\rangle$ into $B_1, \ldots, B_{n-k}$ and obtains $\langle V_j, p_j\rangle$ as the response to each of these operations.
Notice that $p_j$ is the only process that swaps $p_j$ to the identifier field of any object, so no other process applies a \emph{Swap} operation to any object $B_i$ between $C_j$ and the last time $p_j$ applies a \emph{Swap} operation to $B_i$.
Since $p_{j-1}$ decides before $p_j$ decides, this implies that the last $n-k$ \emph{Swap} operations by $p_{j-1}$ happen before $C_j$.
In particular, $p_{j-1}$ decides before $C_j$ in $\alpha$.
Furthermore, $C_1, \ldots, C_{k+1}$ appear in order in $\alpha$.


For all $i \in \{1, \ldots, k+1\}$, define $\ell_i = V_i[x_i]$.
We first show that $\ell_{k+1} \geq \ell_j$ for all $j \in \{1, \ldots, k\}$.
To obtain a contradiction, suppose that $\ell_{k+1} < \ell_j$ for some $j \in \{1, \ldots, k\}$.
Let $j$ be the maximum number for which this is satisfied.
Then for all $j' \in \{j+1, \ldots, k\}$, $\ell_{k+1} \geq \ell_{j'}$.
Since $p_{k+1}$ decides $v_{k+1}$ while its local lap counter is $V_{k+1}$, the condition on line~\ref{ln:decidecond} implies that $\ell_{k+1} \geq V_{k+1}[x_j] + 2$.
Since $\ell_j > \ell_{k+1}$, we have $V_j[v_j] = \ell_j > V_{k+1}[v_j] + 2$.
Hence, $V_j \not\preceq V_{k+1}$.
By Lemma~\ref{lem:manyprocesses}, there are $n-k$ distinct processes different from $p_j$ and $p_{k+1}$ that apply \emph{Swap} operations to different objects between $C_j$ and $C_{k+1}$, and the values of the local lap counters of these processes in $C_{k+1}$ dominate $V_j$.
Since there are only $n-k-1$ processes other than $p_1, \ldots, p_{k+1}$, at least one of these $n-k$ processes must be some $p_{j'} \in \{p_1, \ldots, p_{k+1}\}$.
Note that $p_{j'} \not \in \{p_1, \ldots, p_{j-1}\}$, since $p_1, \ldots, p_{j-1}$ decided before $C_j$.
Hence, $p_{j'} \in \{p_{j+1}, \ldots, p_k\}$.
Since $p_{j'}$ does not change its local lap counter after $C_{j'}$ in $\alpha$, and $C_{j'}$ appears before $C_{k+1}$ in $\alpha$, the local lap counter of $p_{j'}$ in $C_{k+1}$ is $V_{j'}$.
Hence, $V_{j'}$ dominates $V_j$.
In particular, $V_{j'}[x_j] \geq V_j[x_j] = \ell_j$.
Furthermore, $V_{j'}[x_{j'}] = \ell_{j'} \geq V_{j'}[x_j]+2$ by the condition on line~\ref{ln:decidecond}.
Thus, $\ell_{j'} > \ell_j > \ell_{k+1}$.
This contradicts the definition of $j$.


Hence, $\ell_{k+1} \geq \ell_j$ for all $j \in \{1, \ldots, k\}$.
Let $\ell_t =$ max$\{\ell_j\; :\; 1 \leq j \leq k\}$.
By Observation~\ref{obs:atleast2}, $\ell_t \geq 2$.
Therefore, by Observation~\ref{obs:onebyone}, there is a step before $C_{k+1}$ in $\alpha$ by some process $q$ in which $q$ increments component $x_{k+1}$ of its local lap counter from $\ell_t-1$ to $\ell_t$.
Let $V$ be the value of $q$'s local lap counter immediately before taking this step, so $V[x_{k+1}] = \ell_t-1$.
By definition of $\ell_t$, we have $\ell_t \geq \ell_j$ for all $j \in \{1, \ldots, k\}$.
Since $p_j$ decides $x_j$ while it has local lap counter $V_j$, the condition on line~\ref{ln:decidecond} implies that $\ell_j \geq V_j[x_{i}]+2$ for all $i \in \{1, \ldots, k+1\} - \{j\}$.
Hence, $V[x_{k+1}] = \ell_t-1 \geq \ell_j-1 \geq V_j[x_{k+1}]+1$ for all $j \in \{1, \ldots, k\}$.
When process $p_j$ decides, the value of its local lap counter is $V_j$.
Since $V \not\preceq V_j$ for any $j \in \{1, \ldots, k\}$, Observation~\ref{obs:alwaysforward} implies that $q \not\in \{p_1, \ldots, p_k\}$.

By Observation~\ref{obs:total}, the configuration $C$ immediately after $q$ performed line~\ref{ln:setconflict} for the last time before completing the lap for $x_{k+1}$ is $\langle V, q\rangle$-total.
Suppose that $C$ appears before $C_t$ in $\alpha$.
Since $V[x_{k+1}] = \ell_t-1 > V_t[x_{k+1}]$, we have $V \not\preceq V_t$.
Hence, Lemma~\ref{lem:manyprocesses} implies that there are $n-k$ distinct processes different from $q$ and $p_t$ that apply \emph{Swap} operations to different objects between $C$ and $C_t$, and the values of the local lap counters of these processes dominate $V$ in $C_t$.
Since there are only $n-k-1$ processes other than $p_1, \ldots, p_k, q$, at least one of these processes $p_{t'}$ is in $\{p_1, \ldots, p_k\} - \{q, p_t\}$.
Since $V \preceq V_{t'}$, we know that $V_{t'}[x_{k+1}] \geq V[x_{k+1}] = \ell_t-1$.
This contradicts the fact that $\ell_t-1 \geq \ell_{t'}-1 \geq V_{t'}[x_{k+1}]+1$.


Therefore, $C$ appears after $C_t$ in $\alpha$.
Notice that processes only increment a component of their local lap counter with maximal value on line~\ref{ln:inc}.
Since $q$ increments the $x_{k+1}$-th component of its lap counter from $\ell_t-1$ to $\ell_t$, we know that $V[x_{k+1}] \geq V[x_t]$.
Since $V[x_{k+1}] = \ell_t-1 < \ell_t$, this implies that $V[x_t] < \ell_t = V_t[x_t]$.
Hence, $V_t \not\preceq V$.
By Lemma~\ref{lem:manyprocesses}, there are $n-k$ distinct processes different from $q$ and $p_t$ that apply \emph{Swap} operations to different objects between $C_t$ and $C$, and the values of the local lap counters of these processes dominate $V_t$ in $C$.
Since there are only $n-k-1$ processes other than $p_1, \ldots, p_k, q$, at least one of these processes $p_{t'}$ is in $\{p_1, \ldots, p_k\} - \{q, p_t\}$.
Since $V_t \preceq V_{t'}$, we know that $V_{t'}[x_{t}] \geq V_t[x_{t}] = \ell_t$.
This contradicts the fact that $\ell_t-1 \geq \ell_{t'}-1 \geq V_{t'}[x_{t}]+1$.
\end{proof}

\begin{lemma}\label{lem:valid}
Algorithm~\ref{alg:setagreeswap} satisfies validity.
\end{lemma}

\begin{proof}
Consider an execution $\alpha$ in which some process $p$ decides a value $j$.
We need to prove that some process had input $j$.
By Observation~\ref{obs:atleast2}, the $j$-th component of $p$'s local lap counter is at least $2$ when $p$ decides.
Consider the first step in $\alpha$ in which some process $q$ changes the $j$-th component of its local lap counter so that it is greater than $0$.
Then $q$ changes the $j$-th component of its local lap counter on line~\ref{ln:init2}, line~\ref{ln:setmaxlocal}, or line~\ref{ln:inc}.

First, suppose $q$ changes the $j$-th component of its local lap counter on line~\ref{ln:setmaxlocal}.
Then the response of the previous \emph{Swap} operation that $q$ performed on line~\ref{ln:swap} must have had a lap counter $V$ with $V[j] > 0$.
Hence, prior to this step by $q$, the local lap counter of some process was equal to $V$, which is a contradiction.

Now suppose $q$ changes the $j$-th component of its local lap counter on line~\ref{ln:inc}.
Then $q$ changes the $j$-th component of its local lap counter from $0$ to $1$.
Since processes always increment a component of their local lap counter with maximal value on line~\ref{ln:inc}, all of the other components of $q$'s local lap counter are at most $0$ immediately before $q$ increments the $j$-th component.
However, on line~\ref{ln:init2}, $q$ sets some component of its local lap counter to $1$, and it never decreases the value of any component of its local lap counter.
This is a contradiction.

Thus, process $q$ changes the $j$-th component of its local lap counter on line~\ref{ln:init2}.
This implies that $q$ had input $j$.
\end{proof}

\begin{lemma}\label{lem:of}
Algorithm~\ref{alg:setagreeswap} is obstruction-free.
\end{lemma}

\begin{proof}
For any configuration $C$, we will show that a solo execution by $p$ from $C$ contains at most $8(n-k)$ steps.
Suppose that $p$ is poised to apply a \emph{Swap} operation to the object $B_i$ in $C$.
Then $p$ applies exactly one \emph{Swap} operation to each of $B_i, \ldots, B_{n-k}$ before completing the loop on lines~\ref{ln:swaploop}--\ref{ln:endswaploop}.
Afterwards, $p$ either decides and returns, or it begins a new iteration of the loop on lines~\ref{ln:mainloop}--\ref{ln:endmainloop}.
During this loop, $p$ applies exactly one \emph{Swap} operation to every object on lines~\ref{ln:swaploop}--\ref{ln:endswaploop}.
Notice that $p$ may change its local lap counter after swapping one of the objects $B_1, \ldots, B_{i-1}$, but the value of its local lap counter immediately before starting the loop dominates the lap counters stored in $B_i, \ldots, B_{n-k}$.
After this, every object contains $p$ in its identifier field, and the value $V$ of $p$'s local lap counter dominates the lap counter stored in every object.
Hence, after one more iteration of the loop on lines~\ref{ln:mainloop}--\ref{ln:endmainloop}, the resulting configuration $C'$ is $\langle V, p\rangle$-total, and $p$ has either returned or it is poised to access $B_1$.
At this point, $p$ has performed at most $3(n-k)$ steps.

If $p$ has not returned in $C'$, then it swaps every object once before completing a lap.
When this happens, $p$ either returns or it increments a component $x$ of its local lap counter with maximal value to obtain a new local lap counter value $V'$.
In the second case, $V'[x] \geq V'[x']+1$, for all $x' \in \{0, \ldots, m-1\} - \{x\}$.
Process $p$ then swaps $\langle V', p\rangle$ into every object exactly once, and the resulting configuration $C''$ is $\langle V', p\rangle$-total.
In $C''$, process $p$ has either returned or it is poised to access $B_1$.
Process $p$ then swaps every object exactly one more time before completing another lap.
After this, $p$ either returns or increments the $x$-th component of its local lap counter to obtain a new lap counter $V''$, where $V''[x] \geq V''[x']+2$ for all $x' \in \{0, \ldots, m-1\} - \{x\}$.
Finally, $p$ performs exactly two more iterations of the loop on lines~\ref{ln:mainloop}--\ref{ln:endmainloop}, swapping each object exactly two more times before returning.
In total, $p$ performs at most $3(n-k) + 5(n-k) = 8(n-k)$ \emph{Swap} operations in any solo execution before returning.
\end{proof}

Lemmas~\ref{lem:agree}-\ref{lem:of} imply that Algorithm~\ref{alg:setagreeswap} is an $n$-process, $m$-valued, obstruction-free $k$-set agreement algorithm.

\section{Lower Bound on Set Agreement from Swap}\label{sec:unreadablelb}

We now present our lower bound on the number of swap objects required to solve nondeterministic solo-terminating $(k+1)$-valued $k$-set agreement.
Consider an initial configuration $C$ and an execution $\alpha$ from $C$ in which $k$ distinct values are decided.
Let $\mathcal{Q}$ be a set of processes that do not participate in $\alpha$. 
Suppose that, in $C$, the processes in $\mathcal{Q}$ have input $v$, which is different from all of the values decided in $C\alpha$.
If a process $q \in \mathcal{Q}$ starts running in $C\alpha$, then it cannot tell that $k$ different values have been decided until it applies a \emph{Swap} operation to some object that was modified during $\alpha$.
However, when $q$ first swaps some object $B$ that was modified during $\alpha$, process $q$ overwrites the information about $\alpha$ that was stored there.
Hence, if we stop $q$ immediately after it swaps $B$, no other process in $\mathcal{Q}$ can learn anything about $\alpha$ by swapping $B$.
We apply this argument repeatedly, using different processes to overwrite each of the objects that store information about $\alpha$.
This shows that $\alpha$ must have swapped at least $|\mathcal{Q}|$ different objects.
We formalize this argument in the proof of Lemma~\ref{lem:consume}.

\begin{lemma}\label{lem:consume}
	Consider an initial configuration $C$ of a nondeterministic solo-terminating $k$-set agreement algorithm in which a set of processes $\mathcal{Q}$ have the same input $v$.
	Suppose that there is an execution $\alpha$ from $C$ that contains no steps by processes in $\mathcal{Q}$ such that $k$ distinct values different from $v$ are decided in $C\alpha$.
	Then the algorithm uses at least $|\mathcal{Q}|$ swap objects.
\end{lemma}

\begin{proof}
Let $\mathcal{Q} = \{q_1, \ldots, q_{|\mathcal{Q}|}\}$.
Define $\mathcal{Q}_i = \{q_1, \ldots, q_i\}$ for all $i \in \{1, \ldots, |\mathcal{Q}|\}$, and define $\mathcal{Q}_0 = \emptyset$.
Let $D$ be an initial configuration in which all processes have input $v$.
For all $0 \leq i \leq |\mathcal{Q}|$, we show that there is a set of $i$ swap objects $\mathcal{A}_i$ and a pair of $\mathcal{Q}_i$-only executions $\gamma_i$ and $\delta_i$ from $C\alpha$ and $D$, respectively, such that $\textit{value}(B, C\alpha\gamma_i) = \textit{value}(B, D\delta_i)$, for all $B \in \mathcal{A}_i$.
We use induction on $i$.
When $i = 0$, $\gamma_i$ and $\delta_i$ are empty executions, $\mathcal{A}_i = \emptyset$, and the claim is trivially satisfied.

\begin{figure}[h]
\centering
\begin{tikzpicture}[shorten >=1pt,node distance=2.9cm,on grid,auto] 
	\tikzset{every state/.style={minimum size=0.85cm}}

	\node[state]	(C)	{$C$};
	\node[state]	(Calpha) [right=3.3cm of C]	{};
	\node[state]	(Calphagamma)	[right=3.3cm of Calpha]	{};
	\node[state]	(Calphagammaqi+1)	[right=4cm of Calphagamma]	{};
	\node[state] (step)	[right=2cm of Calphagammaqi+1] {};
	
	\node[state] (D) [above=2.6cm of Calpha] {$D$};
	\node[state] (Ddelta) [right=3.3cm of D] {};
	\node[state] (Ddeltatau) [right=4cm of Ddelta] {};
	\node[state] (Ddeltataus) [right=2cm of Ddeltatau] {};
	
	\node (input) [left=1.5cm of D,align=center] {\small all processes \\ \small have input $v$};
	
	\node (decided) [below=1cm of Calpha,align=center] {\small $k$ values different \\ \small from $v$ decided};

	\draw [decorate,decoration={brace,amplitude=6pt}] (3.67,3.6) -- (12.1,3.6) node[pos=0.5,above=7pt] {$\delta_{i+1}$};	
	
	\draw [decorate,decoration={brace,amplitude=6pt}] (3.67,0.6) -- (12.1,0.6) node[pos=0.5,above=7pt] {$\gamma_{i+1}$};

	\node (a1) [rectangle,draw,fill=red,below right=2.7cm and 1cm of C] {};
	\node (a0) [left=0.3cm of a1] {};
	\node (a2) [rectangle,draw,fill=red,right=0.5cm of a1] {};
	\node (a3) [right=0.5cm of a2] {$\ldots$};
	\node (a4) [rectangle,draw,fill=red,right=0.5cm of a3] {};
	\node (a5) [rectangle,draw,right=0.5cm of a4] {};
	\node (a6) [rectangle,draw,right=0.5cm of a5] {};
	\node (a7) [rectangle,draw,right=0.5cm of a6] {};
	\node (a8) [right=0.5cm of a7] {$\ldots$};
	\node (a9) [rectangle,draw,right=0.5cm of a8] {};
	\node (aend) [right=0.3cm of a4] {};
	\draw ([shift={(-1.1cm, 1cm)}]a1) rectangle ([shift={(1.1cm, -1.7cm)}]a9);
	\node (bobj1) [below=1.2cm of a5,align=left] {\small Base objects in $C\alpha\gamma_i$. The base objects \\ \small in $\mathcal{A}_i$ have the same value as in $D\delta_i$.};
	\node (bstar) [below=0.4cm of a5] {$B^\star$};
	\draw [decorate,decoration={brace,amplitude=4pt,raise=0.2cm}] (a0) -- (aend) node[pos=0.5,above=7pt] {$\mathcal{A}_i$};

	\node (b1) [rectangle,draw,fill=red,below right=2.7cm and 1cm of Calphagamma] {};
	\node (b0) [left=0.3cm of b1] {};
	\node (b2) [rectangle,draw,fill=red,right=0.5cm of b1] {};
	\node (b3) [right=0.5cm of b2] {$\ldots$};
	\node (b4) [rectangle,draw,fill=red,right=0.5cm of b3] {};
	\node (b5) [rectangle,draw,fill=red,right=0.5cm of b4] {};
	\node (b6) [rectangle,draw,right=0.5cm of b5] {};
	\node (b7) [rectangle,draw,right=0.5cm of b6] {};
	\node (b8) [right=0.5cm of b7] {$\ldots$};
	\node (b9) [rectangle,draw,right=0.5cm of b8] {};
	\node (bend) [right=0.3cm of b5] {};
	\draw ([shift={(-1.1cm, 1cm)}]b1) rectangle ([shift={(1.1cm, -1.7cm)}]b9);
	\node (bobj2) [below=1.2cm of b5,align=left] {\small Base objects in $C\alpha\gamma_{i+1}$. The base objects \\ \small in $\mathcal{A}_{i+1}$ have the same value as in $D\delta_{i+1}$.};
	\node (bstar2) [below=0.4cm of b5] {$B^\star$};
	\draw [decorate,decoration={brace,amplitude=4pt,raise=0.25cm}] (b0) -- (bend) node[pos=0.5,above=8pt] {$\mathcal{A}_{i+1}$};
	
   	\node (qi+1swap) [above=0.75cm of b5] {\color{red}$s'$};
   	
   	\draw [dashed] (Ddeltatau) -- (Calphagammaqi+1);
   	\node[fill=white,below=1.3cm of Ddeltatau] {$\tau' \widesim{q_{i+1}} \tau$};
   
    \path[->,
	line join=round, decoration={
	    zigzag,
	    segment length=4,
	    amplitude=.9,post=lineto,
	    post length=2pt}]	    
	    
	(C)	edge[decorate] node[below] {no steps by $\mathcal{Q}$}  node[above] {$\alpha$} (Calpha)
	(Calpha)	edge[decorate] node[below] {\small $\mathcal{Q}_i$-only} node[above] {$\gamma_i$} (Calphagamma)
	(Calphagamma)	edge[decorate] node[below] {$q_{i+1}$-only} node[above] {$\tau'$} (Calphagammaqi+1)
	(Calphagammaqi+1)	edge[color=red] node[above] {$s'$} node[below=0.4cm,align=center] {\small $q_{i+1}$ applies \\ \small \emph{Swap} to $B^\star \not\in \mathcal{A}_i$} (step)
	(qi+1swap) edge[color=red] (b5)
	(D) edge[decorate] node[above] {$\delta_i$} node[below] {$\mathcal{Q}_i$-only} (Ddelta)
	(Ddelta) edge[decorate] node[above] {$\tau$} node[below] {$q_{i+1}$-only} (Ddeltatau)
	(Ddeltatau) edge node[above] {$s$} (Ddeltataus);
\end{tikzpicture}

\vspace{0.3cm}

\caption{The construction of $\gamma_{i+1}$ from $\gamma_i$ in the proof of Lemma~\ref{lem:consume}.}\label{fig:consume}
\end{figure}
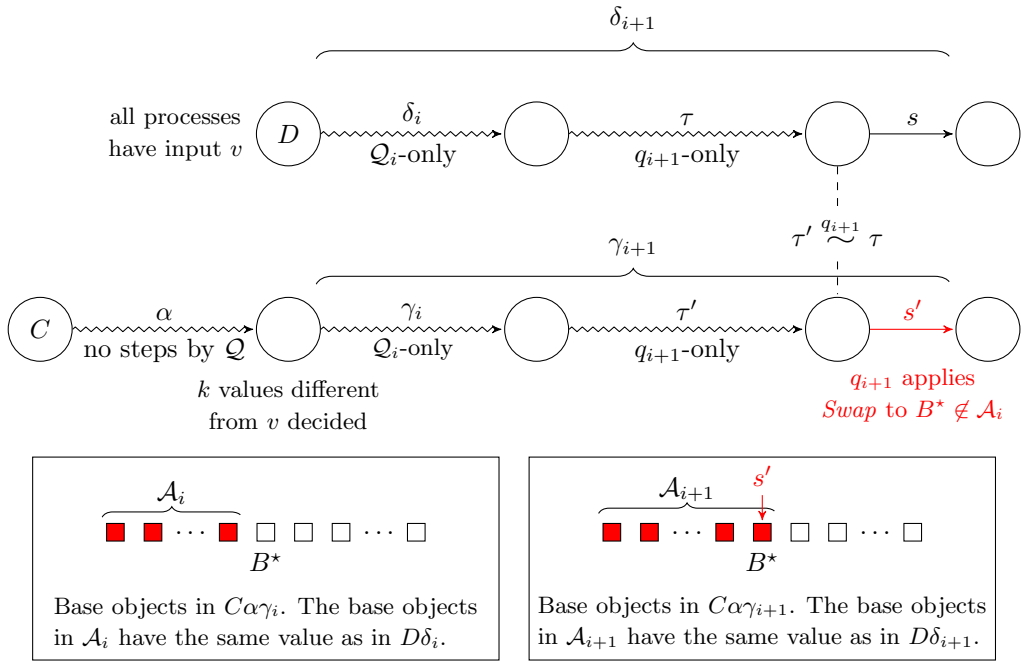

Now suppose the claim holds for some $0 \leq i < |\mathcal{Q}|$.
Notice that $C\alpha\gamma_i \widesim{q_{i+1}} D\delta_i$, since $q_{i+1}$ has input $v$ in both configurations and takes no steps in $\alpha$, $\gamma_i$, or $\delta_i$.
Consider a $q_{i+1}$-only solo-terminating execution $\sigma$ from $D\delta_i$.
By validity, $q_{i+1}$ decides $v$ in $\sigma$.
Let $\tau$ be the longest prefix of $\sigma$ such that $q_{i+1}$ only accesses objects in $\mathcal{A}_i$ during $\tau$.
By the induction hypothesis, $\textit{value}(B, C\alpha\gamma_i) = \textit{value}(B, D\delta_i)$ for all $B \in \mathcal{A}_i$.
Hence, there is a $q_{i+1}$-only execution $\tau'$ from $C\alpha\gamma_i$ such that $\tau' \widesim{q_{i+1}} \tau$.
If $\tau = \sigma$, then $q_{i+1}$ decides $v$ in $\tau'$ and $\tau$.
This contradicts $k$-agreement, since $k$ distinct values different from $v$ are decided in $C\alpha$.
Thus, $\tau$ is a proper prefix of $\sigma$.

Let $s'$ and $s$ be the steps that $q_{i+1}$ is poised to apply in $C\alpha\gamma_i\tau'$ and $D\delta_i\tau$, respectively.
Since $C\alpha\gamma_i \widesim{q_{i+1}} D\delta_i$ and $\tau' \widesim{q_{i+1}} \tau$, we have $C\alpha\gamma_i\tau' \widesim{q_{i+1}} D\delta_i\tau$.
Thus, $q_{i+1}$ performs the same \emph{Swap} operation in $s$ and $s'$.
Furthermore, since $q_{i+1}$ applies the same sequence of operations in $\tau'$ and $\tau$, we have $\textit{value}(B, C\alpha\gamma_i\tau') = \textit{value}(B, D\delta_i\tau)$ for all $B \in \mathcal{A}_i$.
By definition of \emph{Swap}, if $B^\star$ is the object accessed by $s'$ and $s$, then $\textit{value}(B^\star, C\alpha\gamma_i\tau's') = \textit{value}(B^\star, D\delta_i\tau s)$.
By definition of $\tau$, we know that $B^\star \not\in \mathcal{A}_i$.
Hence, taking $\gamma_{i+1} = \gamma_i\tau's'$, $\delta_{i+1} = \delta_i\tau s$, and $\mathcal{A}_{i+1} = \mathcal{A}_i \cup \{B^\star\}$ completes the inductive step.
This is illustrated in Figure~\ref{fig:consume}.

Taking $i = |\mathcal{Q}|$ gives us a set of $|\mathcal{Q}|$ objects $\mathcal{A}_{|\mathcal{Q}|}$.
This completes the proof of the lemma.
\end{proof}

\begin{theorem}\label{thm:swaplb}
	For all $n > k \geq 1$, every nondeterministic solo-terminating, $n$-process $(k+1)$-valued $k$-set agreement algorithm from swap objects uses at least $\lceil\frac{n}{k}\rceil - 1$ objects.
\end{theorem}

\begin{proof}
	Consider such an algorithm for the set of processes $\mathcal{P} = \{p_0, \ldots, p_{n-1}\}$.
	We will use induction on $k$.
	When $k = 1$, consider an initial configuration $C$ of the algorithm in which process $p_0$ has input $0$ and all other processes have input $1$.
	Let $C_0$ be an initial configuration in which all processes have input $0$.
	
	Let $\alpha'$ be a solo-terminating execution by $p_0$ from $C_0$.
	Since $C_0 \widesim{p_0} C$ and all of the base objects have the same values in $C_0$ and $C$, there is a $p_0$-only execution $\alpha$ by $p_0$ from $C$ such that $\alpha' \widesim{p_0} \alpha$.
	Since $p_0$ decides during $\alpha'$, it must decide during $\alpha$ as well.
	By validity, $p_0$ decides $0$ during $\alpha'$, which implies that $p_0$ decides $0$ during $\alpha$.
	Lemma~\ref{lem:consume} (with $\mathcal{Q} = \{p_1, \ldots, p_{n-1}\}$ and $v = 1$) implies that the algorithm uses at least $n-1$ swap objects, which completes the proof of the base case.
	
	Now let $1 < k < n$ and suppose the theorem holds for $k-1$.
	Let $\mathcal{R}$ be some set of $\bigl\lceil \frac{n(k-1))}{k}\bigr\rceil$ processes in $\mathcal{P}$.
	Let $\mathcal{I}$ be the set of all initial configurations in which $\mathcal{R}$'s inputs are in $\{0, \ldots, k-1\}$.
	If, for every initial configuration $C \in \mathcal{I}$ and every $\mathcal{R}$-only execution from $C$, at most $k-1$ different values are decided in $\alpha$, then the algorithm solves nondeterministic solo-terminating $k$-valued $(k-1)$-set agreement among the processes in $\mathcal{R}$.
	Hence, by the induction hypothesis, the algorithm uses at least $\bigl\lceil\frac{|\mathcal{R}|}{k-1}\bigr\rceil - 1 = \lceil\frac{n}{k}\rceil - 1$ swap objects.
	
	Otherwise, there is an initial configuration $C \in \mathcal{I}$ and an $\mathcal{R}$-only execution $\alpha$ from $C$ in which all $k$ values $0, \ldots, k-1$ are decided.
	Notice that $|\mathcal{P} - \mathcal{R}| = n - \bigl\lceil\frac{n(k-1)}{k}\bigr\rceil = \lfloor\frac{n}{k}\rfloor \geq \lceil\frac{n}{k}\rceil - 1$.
	Hence, by Lemma~\ref{lem:consume} (with $\mathcal{Q} = \mathcal{P} - \mathcal{R}$ and $v = k$), the algorithm uses at least $|\mathcal{P} - \mathcal{R}| \geq \lceil\frac{n}{k}\rceil - 1$ swap objects.
	This completes the proof of the theorem.
\end{proof}

Any historyless object that supports only nontrivial operations can be simulated by a single swap object \citep{efr-07}.
Therefore, Theorem~\ref{thm:swaplb} implies that nondeterministic solo-terminating $k$-set agreement cannot be solved with fewer than $\lceil\frac{n}{k}\rceil-1$ historyless objects that support only nontrivial operations.
The proof of Lemma~\ref{lem:consume} exploits a key limitation of nontrivial operations on historyless objects.
In order to learn important information stored in a historyless object that only supports nontrivial operations, a process must overwrite that information.
This prevents other processes from learning the information.
The reason this argument does not work with readable swap objects is because, by using a \emph{Read} operation, a process may be able to see some of the values that have been decided in a configuration without overwriting anything.


Since all obstruction-free and randomized wait-free consensus algorithms are also nondeterministic solo-terminating, Theorem~\ref{thm:swaplb} gives us the following.

\begin{corollary}
	For all $n > k \geq 1$, every obstruction-free or randomized wait-free, $n$-process, $(k+1)$-valued, $k$-set agreement algorithm from swap objects requires at least $\lceil\frac{n}{k}\rceil - 1$ objects.
\end{corollary}

\section{Readable swap objects}\label{sec:readable}

In this section, we consider obstruction-free binary consensus algorithms from readable swap objects with bounded domains.
We first prove that at least $n-2$ readable binary swap objects are needed to solve obstruction-free binary consensus.
Then, we prove that at least $\frac{n-2}{3b+1}$ readable swap objects with domain size $b$ are needed to solve obstruction-free binary consensus.
For $b=2$, our first lower bound of $n-2$ is better than our second lower bound of $\frac{n-2}{7}$.

Throughout this section, let $\mathcal{P} = \{p_0, \ldots, p_{n-3}\}$ be a set of $n-2$ processes and let $\mathcal{Q} = \{q_0, q_1\}$ be a special pair of processes disjoint from $\mathcal{P}$.
For all $0 \leq i \leq n-3$, define $\mathcal{P}_i = \{p_i, \ldots, p_{n-3}\}$ to be all but the first $i$ processes in $\mathcal{P}$.
Define $\mathcal{P}_{n-2} = \emptyset$.
We will first prove some properties that will be useful for obtaining both of our lower bounds.

\begin{observation}\label{obs:initbiv}
	Let $C$ be an initial configuration in which process $q_0$ has input $0$ and process $q_1$ has input $1$.
	Then $\mathcal{Q}$ is bivalent in $C$.
\end{observation}

\begin{proof}
Let $D_0$ and $D_1$ be the initial configurations in which all processes have input $0$ and $1$, respectively.
Then $C \widesim{q_0} D_0$ and $C \widesim{q_1} D_1$.
Furthermore, all of the base objects have the same values in $C$, $D_0$, and $D_1$.
By validity, $q_0$ decides $0$ in its solo-terminating execution from $D_0$, and hence it decides $0$ in its solo-terminating execution from $C$.
Similarly, process $q_1$ decides $1$ in its solo-terminating execution from $C$.
Thus, $\mathcal{Q}$ is bivalent in $C$.
\end{proof}

In the following lemma, we show that if $C$ is a configuration in which $\mathcal{Q}$ is bivalent and some set of processes $S \subseteq \mathcal{P}$ covers a set of objects $\mathcal{B}$, then there is a $\mathcal{Q}$-only execution from $C$ followed by the block swap by $S$ that leads to a configuration in which $\mathcal{Q}$ is bivalent.
This is a generalization of Lemma~3.5 in \citep{z-16}, which proves the same property when the implementation uses registers.

\begin{lemma}\label{lem:extend}
	Let $C$ be a configuration in which $\mathcal{Q}$ is bivalent and a set $S \subseteq \mathcal{P}$ of processes cover a set $\mathcal{B}$ of readable swap objects. Then there is a $\mathcal{Q}$-only execution $\gamma$ from $C$ such that $\mathcal{Q}$ is bivalent in $C\gamma\beta$, where $\beta$ is a block swap by $S$.
\end{lemma}

\begin{proof}
If $\mathcal{Q}$ is bivalent in $C\beta$, then the lemma is satisfied when $\gamma$ is the empty execution.

Now suppose that $\mathcal{Q}$ is $v$-univalent in $C\beta$, for some $v \in \{0, 1\}$.
Since $\mathcal{Q}$ is bivalent in $C$, there is a $\mathcal{Q}$-only execution $\alpha$ from $C$ in which some process in $\mathcal{Q}$ decides $\bar{v}$.
Hence, $\mathcal{Q}$ is $\bar{v}$-univalent in $C\alpha\beta$.

Let $\alpha'$ be some prefix of $\alpha$ such that $\mathcal{Q}$ is $v$-univalent in $C\alpha'\beta$ and $\mathcal{Q}$ is not $v$-univalent in $C\alpha's\beta$, where $s$ is the next step of $\alpha$ by one of the processes in $\mathcal{Q}$.
Without loss of generality, suppose that $q_0$ applies $s$.
There are two cases.

First suppose that $s$ is applied to some object in $\mathcal{B}$.
Then all of the objects have the same values in $C\alpha'\beta$ and $C\alpha's\beta$, and $C\alpha'\beta \stackrel{q_1}{\sim} C\alpha's\beta$.
Since $\mathcal{Q}$ is $v$-univalent in $C\alpha'\beta$, process $q_1$ must decide $v$ in its solo-terminating execution from $C\alpha'\beta$.
Hence, $q_1$ decides $v$ in its solo-terminating execution from $C\alpha's\beta$.
Since $\mathcal{Q}$ is not $v$-univalent in $C\alpha's\beta$ by definition of $\alpha'$, $\mathcal{Q}$ must be bivalent in $C\alpha's\beta$.
	
Otherwise, $s$ is applied to some object outside $\mathcal{B}$.
Then $C\alpha'\beta s = C\alpha's \beta$.
Since $\mathcal{Q}$ is $v$-univalent in $C\alpha'\beta$, process $q_0$ must decide $v$ in its solo-terminating execution from $C\alpha'\beta s$.
Hence, $q_0$ decides $v$ in its solo-terminating execution from $C\alpha's\beta$.
Since $\mathcal{Q}$ is not $v$-univalent in $C\alpha's\beta$ by definition of $\alpha'$, $\mathcal{Q}$ must be bivalent in $C\alpha's\beta$.

In both cases, $\mathcal{Q}$ is bivalent in $C\alpha's \beta$.
Therefore, the lemma is satisfied by $\gamma = \alpha's$.
\end{proof}

The next lemma is a key property that we use to obtain our lower bounds in the next two sections.
Let $C$ be some configuration in which $\mathcal{Q}$ is bivalent, and let $C'$ be a configuration that is indistinguishable from $C$ to a process $p_i$.
Consider $p_i$'s solo-terminating execution $\delta$ from $C'$.
Let $\delta_j$ be the longest prefix of $\delta$ such that, for every prefix $\delta_{j'}$ of $\delta_j$, there is a $(\mathcal{Q} \cup \mathcal{P}_{i+1})$-only execution $\alpha$ from $C$ that is indistinguishable from $\delta_{j'}$ to $p_i$.
Consider the next step $d$ that $p_i$ is poised to take in $C\delta_j$.
Suppose that $p_i$ obtains the response $v$ when it applies $d$ in $C'\delta_j$.
Then for any $(\mathcal{Q} \cup \mathcal{P}_i)$-only execution $\lambda$ from $C\alpha$ in which $p_i$ takes exactly one step, obtaining the response $v$, $\lambda$ and $\delta_jd$ are indistinguishable to $p_i$.
Hence, by definition of $\delta_j$, $\mathcal{Q}$ is univalent in $C\alpha\lambda$.
We formalize this property in the following lemma.

\begin{lemma}\label{lem:solotrick}
Let $p_i \in \mathcal{P}$, let $C$ be a configuration in which $\mathcal{Q}$ is bivalent, let $C'$ be a configuration such that $C \widesim{p_i} C'$, and let $\delta$ be $p_i$'s solo-terminating execution from $C'$.
Suppose $\delta$ consists of $r$ steps and, for all $s \in \{0, \ldots, r\}$, let $\delta_s$ be the prefix of $\delta$ that consists of the first $s$ steps by $p_i$.
Then there is a $j \in \{0, \ldots, r-1\}$ such that:

\begin{enumerate}[label=(\alph*),ref=\alph*]
	\item For all $j' \in \{0, \ldots, j\}$, there is a $(\mathcal{Q} \cup \mathcal{P}_i)$-only execution $\alpha_{j'}$ from $C$ such that $\mathcal{Q}$ is bivalent in $C\alpha_{j'}$ and $\alpha_{j'} \widesim{p_i} \delta_{j'}$.\label{lem:solotrick:prefix}
\end{enumerate}
	
\noindent Consider any $(\mathcal{Q} \cup \mathcal{P}_i)$-only execution $\alpha_j$ from $C$ such that $\mathcal{Q}$ is bivalent in $C\alpha_j$ and $\alpha_j \widesim{p_i} \delta_j$.
Let $d$ be the operation that $p_i$ is poised to apply to the object $B$ in $C'\delta_j$.
Then for every $(\mathcal{Q} \cup \mathcal{P}_{i+1})$-only execution $\lambda'$ from $C\alpha_j$:

\begin{enumerate}[label=(\alph*)]
\setcounter{enumi}{1}
	\item if $\textit{value}(B, C\alpha_j\lambda') = \textit{value}(B, C'\delta_j)$, then $\mathcal{Q}$ is univalent in $C\alpha_j\lambda'd$, and\label{lem:solotrick:cover}

	\item if $\textit{value}(B, C'\delta_j) = \textit{value}(B, C'\delta_j d)$ and, in some configuration of $\lambda'$, the value of $B$ is equal to $\textit{value}(B, C'\delta_j)$, then $\mathcal{Q}$ is univalent in $C\alpha_j\lambda'$.\label{lem:solotrick:nochange}
\end{enumerate}
\end{lemma}

\begin{proof}
Consider any execution $\alpha_r$ from $C$ such that $\alpha_r \widesim{p_i} \delta_r$.
Let $u \in \{0, 1\}$ be the value that $p_i$ has decided in $\delta_r$.
Then $p_i$ has decided $u$ in $\alpha_r$.
By agreement, $\mathcal{Q}$ is $u$-univalent in $C\alpha_r$.
On the other hand, if $\alpha_0 = \delta_0$ is the empty execution from $C$, then $\mathcal{Q}$ is bivalent in $C\alpha_0 = C$ and $\alpha_0 \widesim{p_i} \delta_0$.

Let $j \in \{0, \ldots, r-1\}$ be the minimum value such that $\mathcal{Q}$ is univalent in $C\alpha$ for every execution $\alpha$ from $C$ with $\alpha \widesim{p_i} \delta_{j+1}$.
(Note that there might be no execution $\alpha$ from $C$ with $\alpha \widesim{p_i} \delta_{j+1}$.)
Then for all $j' \in \{0, \ldots, j\}$, there is a $(\mathcal{Q} \cup \mathcal{P}_i)$-only execution $\alpha_{j'}$ from $C$ such that $\mathcal{Q}$ is bivalent in $C\alpha_{j'}$ and $\alpha_{j'} \widesim{p_i} \delta_{j'}$.
This gives us part~(\ref{lem:solotrick:prefix}).

\begin{figure}[h]
\centering
\begin{tikzpicture}[shorten >=1pt,node distance=2.9cm,on grid,auto] 
	\tikzset{every state/.style={minimum size=0.85cm}}

	\node[state]		(C')		{$C'$};
	\node[state]		(C1') [right=2cm of C']	{};
	\node[state]		(C2') [right=2cm of C1'] {};
	\node			(dotstop) [right=1.2cm of C2']	{$\ldots$};
	\node[state]		(Cj') [right=1.2cm of dotstop]		{};
	\node[state]		(Cj1') [right=2cm of Cj']	{};
	
	\node[state,accepting]	(C) [below=2cm of C']		{$C$};
	\node[state,accepting]	(C1) [below=2cm of C1']	{};
	\node[state,accepting]	(C2) [below=2.8cm of C2']	{};
	\node[state,accepting]	(Cj) [below=3.7cm of Cj']		{};
	\node[state]	(Cj1) [below=4.5cm of Cj1']	{};
	
	\draw [dashed] (C') -- (C);
	\draw [dashed] (C1') -- (C1);
	\draw [dashed] (C2') -- (C2);
	\draw [dashed] (Cj') -- (Cj);
	\draw [dashed] (Cj1') -- (Cj1);
	
	\node (C'indistC) [fill=white,below=0.95cm of C'] {\small $C \widesim{p_i} C'$};
	
	\node (indist1) [fill=white,below=0.95cm of C1'] {\small $\alpha_1 \widesim{p_i} \delta_1$};
	
	\node (indist2) [fill=white,below=1.35cm of C2'] {\small $\alpha_2 \widesim{p_i} \delta_2$};
	
	\node (indist3) [fill=white,below=1.8cm of Cj'] {\small $\alpha_j \widesim{p_i} \delta_j$};
	
	\node (explain) [fill=white,above right=1.35cm and 0.4cm of Cj1,align=left] {\small If $\alpha \widesim{p_i} \delta_{j+1}$, then $\mathcal{Q}$ \\ \small is univalent in $C\alpha$};

    \path[->,
	line join=round, decoration={
	    zigzag,
	    segment length=4,
	    amplitude=.9,post=lineto,
	    post length=2pt}]	    
	    
	(C')	 edge node[above] {$\delta_1$} (C1')
	(C1') edge (C2')
	(C2')	edge (dotstop)
	(dotstop) edge (Cj')
	(Cj') edge (Cj1')
	
	(C) edge[decorate] node[above] {$\alpha_1$} (C1)
	(C) edge[decorate,bend right=10] node[above right=-0.1cm and 0.75cm] {$\alpha_2$} (C2)
	(C) edge[decorate,bend right=15] node[above right=-0.3cm and 1.65cm] {$\alpha_j$} (Cj)
	(C) edge[decorate,bend right=18] node[above right=-0.4cm and 3.1cm] {$\alpha$} (Cj1);
	
	\draw [decorate,decoration={brace,amplitude=4pt}] (0.4, 0.5) -- (3.5, 0.5) node[pos=0.5,above=2pt] {$\delta_2$};
	
	\draw [decorate,decoration={brace,amplitude=4pt}] (0.4, 1.1) -- (5.9, 1.1) node[pos=0.5,above=2pt] {$\delta_j$};
	
	\draw [decorate,decoration={brace,amplitude=4pt}] (0.4, 1.7) -- (7.9, 1.7) node[pos=0.5,above=2pt] {$\delta_{j+1}$};
\end{tikzpicture}
\caption{Part~(\ref{lem:solotrick:prefix}) of Lemma~\ref{lem:solotrick}. Nodes with double outlines denote configurations in which $\mathcal{Q}$ is bivalent. The executions $\delta_0$ and $\alpha_0$ are empty and omitted from this diagram.}\label{fig:solotrick}
\end{figure}

Define $v = \textit{value}(B, C'\delta_j)$.
Suppose $\textit{value}(B, C\alpha_j\lambda') = v$.
Then process $p_i$ obtains the same response when it applies $d$ in $C\alpha_j\lambda'$ and $C'\delta_j$.
Hence, $\alpha_j\lambda' d \widesim{p_i} \delta_j d = \delta_{j+1}$.
Let $\alpha = \alpha_j\lambda'd$.
Then $\mathcal{Q}$ is univalent in $C\alpha$ by definition of $j$, which proves property \ref{lem:solotrick:cover}.

Suppose that $\textit{value}(B, C'\delta_j d) = v$.
Then $d$ is either a \emph{Read}$(B)$ or a \emph{Swap}$\bigl(B, v\bigr)$ operation.
Further suppose that, in some configuration of $\lambda'$, the value of $B$ is $v$.
Let $\lambda' = \lambda_1'\lambda_2'$, where $\textit{value}(B, C\alpha_j\lambda_1') = v$.
Then $p_i$ obtains the same response when it applies $d$ in $C\alpha_j\lambda_1'$ and $C'\delta_j$.
Thus, $\alpha_j\lambda_1' d \widesim{p_i} \delta_j d$.
Since $d$ is either a \emph{Read}$(B)$ or a \emph{Swap}$\bigl(B, v\bigr)$ operation, the value of $B$ is the same in $C\alpha_j\lambda_1'$ and $C\alpha_j\lambda_1' d$.
Furthermore, since $d$ is a single step that accesses $B$, all of the other objects also have the same values in $C\alpha_j\lambda_1'$ and $C\alpha_j\lambda_1' d$.
Since $p_{i} \not\in (\mathcal{Q} \cup \mathcal{P}_{i+1})$, the configurations $C\alpha_j\lambda_1'$ and $C\alpha_j\lambda_1' d$ are indistinguishable to $\mathcal{Q} \cup \mathcal{P}_{i+1}$.
Then there is a $(\mathcal{Q} \cup \mathcal{P}_{i+1})$-only execution $\lambda_2^\star$ from $C\alpha_j\lambda_1'd$ such that $\lambda_2'$ and $\lambda_2^\star$ are indistinguishable to $\mathcal{Q} \cup \mathcal{P}_{i+1}$.
Since all of the objects have the same values in $C\alpha_j\lambda_1'$ and $C\alpha_j\lambda_1'd$, all of the objects also have the same values in $C\alpha_j\lambda_1'\lambda_2'$ and $C\alpha_j\lambda_1'd\lambda_2^\star$.
Furthermore, $C\alpha_j\lambda' = C\alpha_j\lambda_1'\lambda_2'$ is indistinguishable from $C\alpha_j\lambda_1' d\lambda_2^\star$ to $\mathcal{Q} \cup \mathcal{P}_{i+1}$.
Since $p_{i}$ takes no steps in $\lambda_2^\star$, we know that $\alpha_j\lambda_1' d \widesim{p_{i}} \alpha_j\lambda_1' d\lambda_2^\star$, and therefore $\alpha_j\lambda_1' d\lambda_2^\star \widesim{p_{i}} \delta_j d = \delta_{j+1}$.
Let $\alpha = \alpha_j\lambda_1' d\lambda_2^\star$.
Then $\mathcal{Q}$ is univalent in $C\alpha$ by definition of $j$.
Hence, $\mathcal{Q}$ is univalent in $C\alpha_j\lambda'$ as well.
This concludes the proof of property \ref{lem:solotrick:nochange}.

\begin{figure}
\centering
\begin{minipage}{.49\textwidth}
\centering
\begin{tikzpicture}[shorten >=1pt,node distance=2.9cm,on grid,auto] 
	\tikzset{every state/.style={minimum size=0.85cm}}
	
	\node[state,accepting]	(C) 		{$C$};
	\node[state,accepting]	(Cj) [right=1.5cm of C]		{};
	\node[state]	(lambda) [right=1.5cm of Cj]	{};
	\node[state] (step) [right=1.5cm of lambda] {};
	
	\node[state]		(C')	 [above=2.5cm of C]	{$C'$};
	\node[state]		(Cj') [above=2.5cm of lambda]		{};

	\node (value) [above=0.9cm of Cj',align=center] {\small $B$ has \\ \small value $v$};
	
	\node (value2) [above=0.9cm of lambda,align=center] {\small $B$ has \\ \small value $v$};
	
	\node (univ) [below=0.9cm of step,align=center] {\small $\mathcal{Q}$ is \\ \small univalent};

    \path[->,
	line join=round, decoration={
	    zigzag,
	    segment length=4,
	    amplitude=.9,post=lineto,
	    post length=2pt}]
	    
	(C')	 edge[decorate] node[above] {$\delta_j$} (Cj')
	
	(C) edge[decorate] node[above] {$\alpha_j$} (Cj)
	(Cj) edge[decorate] node[above] {$\lambda'$} (lambda)
	(lambda) edge node[above] {$d$} (step);
\end{tikzpicture}
\captionof{figure}{Part~\ref{lem:solotrick:cover} of Lemma~\ref{lem:solotrick}.}\label{fig:partbi}
\end{minipage}
\begin{minipage}{.49\textwidth}
\centering
\begin{tikzpicture}[shorten >=1pt,node distance=2.9cm,on grid,auto] 
	\tikzset{every state/.style={minimum size=0.85cm}}
	
	\node[state,accepting]	(C) 		{$C$};
	\node[state,accepting]	(Cj) [right=1.5cm of C]		{};
	\node[state]	(lambda1) [right=1.5cm of Cj]	{};
	\node[state] (lambda2) [right=1.5cm of lambda1] {};
	
	\node[state]		(C')	 [above=2.5cm of C]	{$C'$};
	\node[state]		(Cj') [above=2.5cm of lambda1]		{};
	\node[state]		(Cj1') [right=1.5cm of Cj'] {};

	\draw [decorate,decoration={brace,amplitude=4pt}] (2.6, 3) -- (4.9, 3) node[pos=0.5,above=4pt,align=center] {\small $B$ has \\ \small value $v$};
	
	\node (value) [above=0.9cm of lambda1,align=center] {\small $B$ has \\ \small value $v$};
	
	\node (univ) [below=0.9cm of lambda2,align=center] {\small $\mathcal{Q}$ is \\ \small univalent};

    \path[->,
	line join=round, decoration={
	    zigzag,
	    segment length=4,
	    amplitude=.9,post=lineto,
	    post length=2pt}]	    
	    
	(C')	 edge[decorate] node[above] {$\delta_j$} (Cj')
	
	(C) edge[decorate] node[above] {$\alpha_j$} (Cj)
	(Cj) edge[decorate] (lambda1)
	(lambda1) edge[decorate] (lambda2)
	(Cj') edge node[above] {$d$} (Cj1');
	
	\draw [decorate,decoration={brace,amplitude=4pt}] (4.05, -0.45) -- (1.9, -0.45) node[pos=0.5,below=2pt] {$\lambda'$};
\end{tikzpicture}
\captionof{figure}{Part~\ref{lem:solotrick:nochange} of Lemma~\ref{lem:solotrick}.}\label{fig:partbii}
\end{minipage}
\end{figure}
\end{proof}

\subsection{Readable binary swap objects}\label{sec:binary}

We will now prove that at least $n-2$ readable binary swap objects are needed to solve obstruction-free binary consensus.
We begin with a simple observation.

\begin{observation}\label{obs:nochange}
	Let $C$ be a configuration in which a set $S \subseteq \mathcal{P}$ of processes covers a set $\mathcal{B}$ of objects.
	If $\mathcal{Q}$ is bivalent in $C\beta$, where $\beta$ is the block swap by $S$, and for every $B \in \mathcal{B}$ we have $\textit{value}(B, C) = \textit{value}(B, C\beta)$, then $\mathcal{Q}$ is bivalent in $C$.
\end{observation}

\begin{proof}
Since $S$ only applies \emph{Swap} operations to objects in $\mathcal{B}$ during $\beta$, we know that all of the objects outside $\mathcal{B}$ have the same values in $C$ and $C\beta$.
Hence, all of the objects have the same values in $C$ and $C\beta$.
Furthermore, since no process in $\mathcal{Q}$ takes any steps during $\beta$, we have $C \widesim{\mathcal{Q}} C\beta$.
Therefore, $\mathcal{Q}$ is bivalent in $C$.
\end{proof}

Consider an initial configuration $C_0$ of the binary consensus algorithm in which $q_0$ has input $0$ and $q_1$ has input $1$.
Then by Observation~\ref{obs:initbiv}, $\mathcal{Q}$ is bivalent in $C_0$.
Let $\delta$ be $p_{0}$'s solo-terminating execution from $C_0$.
Suppose that $\delta$ consists of $r$ steps and, for all $s \in \{0, \ldots, r\}$, let $\delta_s$ be the prefix of $\delta$ that consists of the first $s$ steps by $p_0$.
Let $j$ be the value satisfies the conditions of Lemma~\ref{lem:solotrick} (with $C = C' = C_0$ and $i = 0$).
Then there is an execution $\alpha_j$ from $C_0$ such that $\mathcal{Q}$ is bivalent in $C_0\alpha_j$ and $\alpha_j \widesim{p_0} \delta_j$.

Let $d$ be the operation that $p_0$ is poised to apply to the object $B$ in $C_0\delta_j$.
First suppose that $\textit{value}(B, C_0\delta_j) = \textit{value}(B, C_0\delta_j d)$.
Consider any $(\mathcal{Q} \cup \mathcal{P}_i)$-only execution $\lambda$ from $C_0\alpha_j$ in which the value of $B$ changes.
Since $B$ is a 1-bit object, there is a configuration of $\lambda$ in which the value of $B$ is $\emph{value}(B, C_0\delta_j)$.
Then by part \ref{lem:solotrick:nochange} of Lemma~\ref{lem:solotrick}, it follows that $\mathcal{Q}$ is univalent in $C_0\alpha_j\lambda$.

Now suppose that $\textit{value}(B, C_0\delta_j) \neq \textit{value}(B, C_0\delta_j d)$.
Consider any $(\mathcal{Q} \cup \mathcal{P}_i)$-only execution $\lambda$ from $C_0\alpha_j$ such that $d$ changes the value of $B$ when applied in $C_0\alpha_j\lambda$.
Then \emph{value}$(B, C_0\alpha_j\lambda) = \textit{value}(B, C_0\delta_j)$.
From part \ref{lem:solotrick:cover} of Lemma~\ref{lem:solotrick}, it follows that $\mathcal{Q}$ is univalent in $C_0\alpha_j\lambda d$.

\medskip

In the next lemma, we generalize the argument above to construct, for all $i \in \{0, \ldots, n-2\}$, a configuration $C_i$ in which $\mathcal{Q}$ is bivalent and two disjoint sets of objects $\mathcal{X}_i$ and $\mathcal{Y}_i$ that satisfy the following properties.
If the value of any object in $\mathcal{X}_i$ changes during a $(\mathcal{Q} \cup \mathcal{P}_i)$-only execution from $C_i$, then $\mathcal{Q}$ is univalent in the resulting configuration.
The objects in $\mathcal{Y}_i$ are covered by a set of processes $S \subseteq \mathcal{P} - \mathcal{P}_i$.
If the block swap by $S$ changes the value of some object in $\mathcal{Y}_i$ following some $(\mathcal{Q} \cup \mathcal{P}_i)$-only execution from $C_i$, then $\mathcal{Q}$ is univalent in the resulting configuration.
Furthermore, we have $|\mathcal{X}_i \cup \mathcal{Y}_i| = i$.
Taking $i = n-2$, this allows us to obtain our lower bound.
We show how to construct these sets of objects in the following lemma.

\begin{lemma}\label{lem:maintech}
	For all $i \in \{0, \ldots, n-2\}$, there is a configuration $C_i$, two disjoint sets of objects $\mathcal{X}_i$ and $\mathcal{Y}_i$ such that $|\mathcal{X}_i \cup \mathcal{Y}_i| = i$, and a set of $|\mathcal{Y}_i|$ processes $S_i \subseteq \mathcal{P} - \mathcal{P}_{i}$ (among the first $i$ processes in $\mathcal{P}$) that satisfy the following properties for every $(\mathcal{Q} \cup \mathcal{P}_i)$-only execution $\lambda$ from $C_i$:
	
\begin{enumerate}[label=(\alph*),ref=\alph*]
	
		\item $\mathcal{Q}$ is bivalent in $C_i$,\label{propBiv}
		
		\item $S_{i}$ covers $\mathcal{Y}_{i}$ in $C_i$,\label{propCover}	
		
		\item if the value of any object in $\mathcal{X}_i$ is changed at any point during $\lambda$, then $\mathcal{Q}$ is univalent in $C_i\lambda$, and		\label{nochange}

		\item if the block swap $\beta_i$ by $S_i$ changes the value of some object in $\mathcal{Y}_i$ when applied in $C_i\lambda$, then $\mathcal{Q}$ is univalent in $C_i\lambda\beta_i$.\label{blockswap}

	\end{enumerate}
\end{lemma}

\begin{proof}
We use induction on $i$.
Let $C_0$ be the bivalent initial configuration defined earlier in which $q_0$ has input $0$ and $q_1$ has input $1$.
This gives us property~(\ref{propBiv}).
Let $\mathcal{X}_0 = \mathcal{Y}_{0} = S_0 = \emptyset$.
Since $\mathcal{X}_0$, $\mathcal{Y}_{0}$, and $S_0$ are empty, properties (\ref{propCover}), (\ref{nochange}), and (\ref{blockswap}) hold vacuously.
This proves the base case.

\smallskip

Now suppose that the lemma holds for some $i \in \{0, \ldots, n-3\}$.
Then there is a configuration $C_i$, two disjoint sets of objects $\mathcal{X}_i$ and $\mathcal{Y}_i$ such that $|\mathcal{X}_i \cup \mathcal{Y}_i| = i$, and a set of $|\mathcal{Y}_i|$ processes $S_i \subseteq \mathcal{P} - \mathcal{P}_i$ that satisfy properties~(\ref{propBiv})--(\ref{blockswap}).
By Lemma~\ref{lem:extend} (with $C = C_i$ and $S = S_i$), there is a $\mathcal{Q}$-only execution $\gamma$ from $C_i$ such that $\mathcal{Q}$ is bivalent in $C_i\gamma\beta_i$, where $\beta_i$ is the block swap by $S_i$.
Property~(\ref{blockswap}) (with $\lambda = \gamma$) implies that, if the block swap $\beta_i$ changes the value of some object in $\mathcal{Y}_i$ when applied in $C_i\gamma$, then $\mathcal{Q}$ is univalent in $C_i\gamma\beta_i$.
Since $\mathcal{Q}$ is bivalent in $C_i\gamma\beta_i$, the block swap $\beta_i$ does not change the value of any object in $\mathcal{Y}_i$ when it is applied in $C_i\gamma$.
In other words, $\textit{value}(Y, C_i\gamma) = \textit{value}(Y, C_i\gamma\beta_i)$ for all $Y \in \mathcal{Y}_i$.
Hence, by Observation~\ref{obs:nochange}, $\mathcal{Q}$ is bivalent in $C_i\gamma$.

Let $\delta$ be $p_{i}$'s solo-terminating execution from $C_i\gamma$.
Suppose that $\delta$ consists of $r$ steps and, for all $s \in \{0, \ldots, r\}$, let $\delta_s$ be the prefix of $\delta$ that consists of the first $s$ steps by $p_i$.
Let $j \in \{0, \ldots, r-1\}$ be the value that satisfies the conditions of Lemma~\ref{lem:solotrick} (with $C = C' = C_i\gamma$).
Let $v = \textit{value}(B, C_i\gamma\delta_j)$.

The following claim is important for our construction.

\begin{claim}\label{claim:BinY}
	If $B \in \mathcal{Y}_i$, then $\textit{value}(B, C_i\gamma\beta_i) = v$.
\end{claim}

\begin{proof}[Proof of Claim~\ref{claim:BinY}]
To obtain a contradiction, suppose $B \in \mathcal{Y}_i$ and $\textit{value}(B, C_i\gamma\beta_i) = \bar{v}$.
Recall that, for all objects $Y \in \mathcal{Y}_i$, $\textit{value}(Y, C_i\gamma) = \textit{value}(Y, C_i\gamma\beta_i)$.
In particular, $B \in \mathcal{Y}_i$, so $\textit{value}(B, C_i\gamma) = \textit{value}(B, C_i\gamma\beta_i) = \bar{v}$.
Since $\textit{value}(B, C_i\gamma\delta_j) = v \neq \textit{value}(B, C_i\gamma)$ and $\delta_j$ is a $p_i$-only execution, process $p_i$ must apply \emph{Swap}$(B, v)$ during $\delta_j$.

Suppose that $p_i$ is poised to apply \emph{Swap}$(B, v)$ in $C_i\gamma\delta_t$, where $0 \leq t \leq j-1$.
Lemma~\ref{lem:solotrick}(\ref{lem:solotrick:prefix}) implies that there is a $(\mathcal{Q} \cup \mathcal{P}_{i})$-only execution $\alpha_{t}$ from $C_i\gamma$ such that $\mathcal{Q}$ is bivalent in $C_i\gamma\alpha_t$ and $\alpha_t \widesim{p_{i}} \delta_t$.
Hence, $p_{i}$ is poised to apply \emph{Swap}$(B, v)$ in $C_i\gamma\alpha_t$.
By Lemma~\ref{lem:extend} (with $C = C_i\gamma\alpha_t$ and $S = S_i$), there is a $\mathcal{Q}$-only execution $\gamma'$ from $C_i\gamma\alpha_t$ such that $\mathcal{Q}$ is bivalent in $C_i\gamma\alpha_t\gamma'\beta_i$.

Since $B \in \mathcal{Y}_i$, there is a process in $S_i$ that covers $B$ in $C_i$ by property~(\ref{propCover}).
The execution $\gamma$ is $\mathcal{Q}$-only, so this process covers $B$ in $C_i\gamma$ as well.
Since $\textit{value}(B, C_i\gamma\beta_i) = \bar{v}$, we know that the process in $S_i$ that covers $B$ in $C_i\gamma$ applies a \emph{Swap}$(B, \bar{v})$ operation in $\beta_i$.
Let $d_t$ be the \emph{Swap}$(B, v)$ operation that $p_{i}$ is poised to apply in $C_i\gamma\alpha_t$.
The execution $\gamma'$ is $\mathcal{Q}$-only, so $p_i$ is still poised to apply $d_t$ in $C_i\gamma\alpha_t\gamma'$.
Since $d_t$ is applied to $B$, it is overwritten by the block swap $\beta_i$, so $\textit{value}(B, C_i\gamma\alpha_t\gamma'd_t\beta_i) = \textit{value}(B, C_i\gamma\alpha_t\gamma'\beta_i)$.
All of the other base objects also have the same values in $C_i\gamma\alpha_t\gamma'd_t\beta_i$ and $C_i\gamma\alpha_t\gamma'\beta_i$.
Hence, $\mathcal{Q}$ is bivalent in $C_i\gamma\alpha_t\gamma'd_t\beta_i$.
However, since $\textit{value}(B, C_i\gamma\alpha_t\gamma'd_t) = v \neq \bar{v} = \textit{value}(B, C_i\gamma\alpha_t\gamma'd_t\beta_i)$, property (\ref{blockswap}) (with $\lambda = \gamma\alpha_t\gamma'd_t$) implies that $\mathcal{Q}$ is univalent in $C_i\gamma\alpha_t\gamma'd_t\beta_i$.
This is a contradiction, which completes the proof of Claim~\ref{claim:BinY}.
\end{proof}

\setlength{\textfloatsep}{19.0pt plus 2.0pt minus 4.0pt}

\begin{figure}
\centering
\begin{tikzpicture}[shorten >=1pt,node distance=2.9cm,on grid,auto] 
	\tikzset{every state/.style={minimum size=1.1cm}}

	\node[state,accepting]	(Ci)	{$C_i$};
	\node[state,accepting]	(gamma) [right=2.4cm of Ci]	{};
	\node[state,accepting]	(alpha)	[right=5cm of gamma]	{$C_{i+1}$};
	\node[state]	(delta)	[above=2.5cm of alpha]	{};
	\node[state,accepting]	(beta)	[above=2.5cm of gamma] {};
	
	\node	(indist)	[below=1.2cm of delta]	{$\delta_j \widesim{p_i} \alpha_j$};

	\node (value) [right=1.2cm of delta,align=center] {\small $B$ has \\ \small value $v$};

    \path[->,
	line join=round, decoration={
	    zigzag,
	    segment length=4,
	    amplitude=.9,post=lineto,
	    post length=2pt}]	    
	    
	(Ci)	edge[decorate] node[below] {$\mathcal{Q}$-only} node[above] {$\gamma$} (gamma)
	(gamma)	edge[decorate] node[below] {\small $(\mathcal{Q} \cup \mathcal{P}_i)$-only}  node[above] {$\alpha_j$} (alpha)
	(gamma)	edge[decorate,bend left=18] node[above left] {$\delta_j$}  node[below right=-0.05cm and -0.05cm] {\small $p_i$-only}  (delta)
	(gamma) edge node[left,align=right] {\small Block swap \\ \small $\beta_i$ by $S_i$} (beta);
\end{tikzpicture}
\caption{The construction of $C_{i+1}$ from $C_i$ in the proof of Lemma~\ref{lem:maintech}. Nodes with double outlines denote configurations in which $\mathcal{Q}$ is bivalent.}\label{fig:binlemma}
\vspace{-0.5cm}
\end{figure}
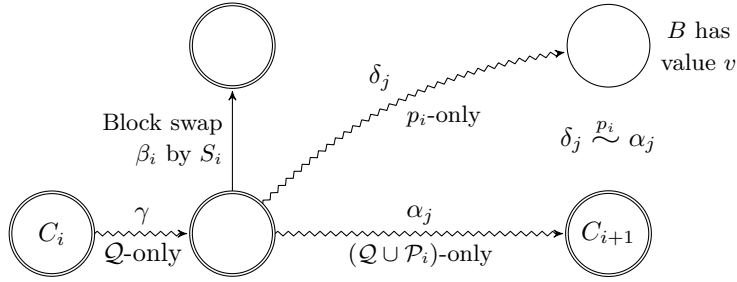

We now proceed with the inductive step.
Lemma~\ref{lem:solotrick}(\ref{lem:solotrick:prefix}) implies that there is a $(\mathcal{Q} \cup \mathcal{P}_i)$-only execution $\alpha_j$ from $C_i\gamma$ such that $\mathcal{Q}$ is bivalent in $C_i\gamma\alpha_j$ and $\alpha_j \widesim{p_i} \delta_j$.
Define $C_{i+1} = C_i\gamma\alpha_j$.
This is illustrated in Figure~\ref{fig:binlemma}.
Since $\mathcal{Q}$ is bivalent in $C_i\gamma\alpha_j = C_{i+1}$, this gives us property (\ref{propBiv}) for $i+1$.

Let $d$ be the operation that $p_i$ is poised to apply to the object $B$ in $C_i\gamma\delta_j$.
Recall that $v = \textit{value}(B, C_i\gamma\delta_j)$.
There are two cases depending on the value of $B$ in configuration $C_i\gamma\delta_jd$.

\smallskip

	\textbf{Case 1:} $\textit{value}(B, C_i\gamma\delta_j d) = v$.
	In this case, since $d$ does not change the value of $B$, $d$ could be either a \emph{Read} or a \emph{Swap}$(B, v)$ operation.
	We first show that $B \not\in \mathcal{X}_i$.
	Lemma~\ref{lem:solotrick}\ref{lem:solotrick:nochange} (with $C = C' = C_i\gamma$ and $\lambda'$ empty) implies that, if $\textit{value}(B, C_i\gamma\alpha_j) = \textit{value}(B, C_i\gamma\delta_j)$, then $\mathcal{Q}$ is univalent in $C_i\gamma\alpha_j$.
	Since $\mathcal{Q}$ is bivalent in $C_i\gamma\alpha_j$, $\textit{value}(B, C_i\gamma\alpha_j) \neq \textit{value}(B, C_i\gamma\delta_j) = v$.
	This implies that $\textit{value}(B, C_i\gamma\alpha_j) = \bar{v}$.
	Since $\textit{value}(B, C_i\gamma\delta_j) = v$, it must be the case that either $\textit{value}(B, C_i\gamma) = v$ or $p_i$ applies a \emph{Swap}$(B, v)$ operation that changed the value of $B$ from $\bar{v}$ to $v$ during $\delta_j$.
	In the second case, since $\delta_j$ and $\alpha_j$ are indistinguishable to $p_i$, process $p_i$ changes the value of $B$ from $\bar{v}$ to $v$ during $\alpha_j$ as well.
	Hence, in either case, the value of $B$ changes during $\alpha_j$.
	If $B \in \mathcal{X}_i$, then property (\ref{nochange}) (with $\lambda = \gamma\alpha_j$) implies that $\mathcal{Q}$ is univalent in $C_i\gamma\alpha_j$.
	This is a contradiction, since $\mathcal{Q}$ is bivalent in $C_i\gamma\alpha_j$.
	Hence, $B \not\in \mathcal{X}_i$.

	Now we show that $B \not\in \mathcal{Y}_i$.
	To obtain a contradiction, suppose $B \in \mathcal{Y}_i$.
	By Claim~\ref{claim:BinY}, $\textit{value}(B, C_i\gamma\beta_i) = v$.
	Then the process in $S_i$ that covers $B$ applies \emph{Swap}$(B, v)$ in $\beta_i$.
	By Lemma~\ref{lem:extend} (with $C = C_i\gamma\alpha_j$ and $S = S_i$), there is a $\mathcal{Q}$-only execution $\gamma''$ from $C_i\gamma\alpha_j$ such that $\mathcal{Q}$ is bivalent in $C_i\gamma\alpha_j\gamma''\beta_i$.
	By property (\ref{blockswap}) (with $\lambda = \gamma\alpha_j\gamma''$), if the block swap $\beta_i$ changes the value of some object in $\mathcal{Y}_i$ when it is applied in $C_i\gamma\alpha_j\gamma''$, then $\mathcal{Q}$ is univalent in $C_i\gamma\alpha_j\gamma''\beta_i$.
	Since $\mathcal{Q}$ is bivalent in $C_i\gamma\alpha_j\gamma''\beta_i$, the block swap $\beta_i$ does not change the value of any object in $\mathcal{Y}_i$ when it is applied in $C_i\gamma\alpha_j\gamma''$.
	In particular, since $B \in \mathcal{Y}_i$, $\textit{value}(B, C_i\gamma\alpha_j\gamma'') = \textit{value}(B, C_i\gamma\alpha_j\gamma''\beta_i) = v$.
	By Observation~\ref{obs:nochange}, $\mathcal{Q}$ is bivalent in $C_i\gamma\alpha_j\gamma''$.
	Since $B$ has value $v$ in $C_i\gamma\alpha_j\gamma''$, Lemma~\ref{lem:solotrick}\ref{lem:solotrick:nochange} (with $C = C' = C_i\gamma$ and $\lambda' = \gamma''$) implies that $\mathcal{Q}$ is univalent in $C_i\gamma\alpha_j\gamma''$.
	This is a contradiction.
	Therefore, $B \not\in \mathcal{Y}_i$.
	
	\smallskip
	
	We have shown that $B \not\in \mathcal{X}_i \cup \mathcal{Y}_i$.
	Define $\mathcal{X}_{i+1} = \mathcal{X}_i \cup \{B\}$, $\mathcal{Y}_{i+1} = \mathcal{Y}_i$, $S_{i+1} = S_i$, and $\beta_{i+1} = \beta_i$.
	Since none of the processes in $S_{i+1} = S_i$ take any steps during $\gamma\alpha_j$, they still cover $\mathcal{Y}_i$ in $C_i\gamma\alpha_j = C_{i+1}$, which proves property (\ref{propCover}) for $i+1$.
	
	\smallskip
	
	Let $\lambda'$ be a $(\mathcal{Q} \cup \mathcal{P}_{i+1})$-only execution from $C_{i+1}$.	
	 Then $\gamma\alpha_j\lambda'$ is a $(\mathcal{Q} \cup \mathcal{P}_{i})$-only execution from $C_i$.
	Hence, by property (\ref{nochange}) (with $\lambda = \gamma\alpha_j\lambda'$), if the value of some object in $\mathcal{X}_i$ changes during $\gamma\alpha_j\lambda'$, then $\mathcal{Q}$ is univalent in $C_i\gamma\alpha_j\lambda' = C_{i+1}\lambda'$.
	Furthermore, by Lemma~\ref{lem:solotrick}\ref{lem:solotrick:nochange} (with $C = C' = C_i\gamma$), if the value of $B$ is equal to $\textit{value}(B, C_i\gamma\delta_j) = v$ at any point during $\lambda'$, then $\mathcal{Q}$ is univalent in $C_i\gamma\alpha_j\lambda' = C_{i+1}\lambda'$.
	If the value of $B$ changes during $\lambda'$, then, since $B$ is a 1-bit object, its value is equal to $v$ at some point during $\lambda'$.
	This implies that if the value of $B$ changes at any point during $\lambda'$, then $\mathcal{Q}$ is univalent in $C_i\gamma\alpha_j\lambda' = C_{i+1}\lambda'$.
	This proves property (\ref{nochange}) for $i+1$.
	
	\smallskip

	Since $\gamma\alpha_j\lambda'$ is a $(\mathcal{Q} \cup \mathcal{P}_i)$-only execution, property (\ref{blockswap}) (with $\lambda = \gamma\alpha_j\lambda'$) implies that, if $\beta_i$ changes the value of an object in $\mathcal{Y}_i = \mathcal{Y}_{i+1}$ when applied in $C_i\gamma\alpha_j\lambda'$, then $\mathcal{Q}$ is univalent in $C_i\gamma\alpha_j\lambda'\beta_i = C_{i+1}\lambda'\beta_i$.
	This proves property (\ref{blockswap}) for $i+1$.

	\medskip

	\textbf{Case 2:} $\textit{value}(B, C_i\gamma\delta_j d) = \bar{v}$.
	We first show that $B \not\in \mathcal{X}_i$.
	To obtain a contradiction, suppose that $B \in \mathcal{X}_i$.
	Since $d$ changes the value of $B$ when applied in $C_i\gamma\delta_j$, we know that $d$ is a \emph{Swap}$(B, \bar{v})$ operation.
	By Lemma~\ref{lem:extend} (with $C = C_i\gamma\alpha_j$ and  $S = \{p_{i}\}$), there is a $\mathcal{Q}$-only execution $\gamma_{i}$ from $C_i\gamma\alpha_j$ such that $\mathcal{Q}$ is bivalent in $C_i\gamma\alpha_j\gamma_{i} d$.
	Property (\ref{nochange}) (with $\lambda = \gamma\alpha_j\gamma_id$) implies that, if the value of $B$ is changed at any point during $\gamma\alpha_j\gamma_id$, then $\mathcal{Q}$ is univalent in $C_i\gamma\alpha_j\gamma_id$.
	Since $\mathcal{Q}$ is bivalent in $C_i\gamma\alpha_j\gamma_id$, the value of $B$ is not changed during $\gamma\alpha_j\gamma_{i}d$.
	Thus, $\textit{value}(B, C_i\gamma) = \textit{value}(B, C_i\gamma\alpha_j\gamma_{i}d) = \bar{v}$.
	Since $\textit{value}(B, C_i\gamma\delta_j) = v$, process $p_i$ applies a \emph{Swap}$(B, v)$ operation during $\delta_j$.
	Since $\alpha_j \widesim{p_i} \delta_j$, process $p_i$ applies a \emph{Swap}$(B, v)$ operation during $\alpha_j$ as well.
	This implies that the value of $B$ changes during $\alpha_j$.
	But this is a contradiction, since the value of $B$ does not change during $\gamma\alpha_j\gamma_id$.
	Hence, $B \not\in \mathcal{X}_i$.

	Now we show that $B \not\in \mathcal{Y}_i$.
	To obtain a contradiction, suppose $B \in \mathcal{Y}_i$.
	By Claim~\ref{claim:BinY}, $\textit{value}(B, C_i\gamma\beta_i) = v$.
	Then the process in $S_i$ that covers $B$ applies \emph{Swap}$(B, v)$ in $\beta_i$.
	By Lemma~\ref{lem:extend} (with $C = C_i\gamma\alpha_j$ and $S = S_i$), there is a $\mathcal{Q}$-only execution $\gamma''$ from $C_i\gamma\alpha_j$ such that $\mathcal{Q}$ is bivalent in $C_i\gamma\alpha_j\gamma''\beta_i$.
	Consider the configuration $C_i\gamma\alpha_j\gamma''d\beta_i$.
	Since $d$ is applied to $B \in \mathcal{Y}_i$, it is overwritten by $\beta_i$.
	That is, $\textit{value}(B, C_i\gamma\alpha_j\gamma''d\beta_i) = \textit{value}(B, C_i\gamma\alpha_j\gamma''\beta_i)$.
	Since $\mathcal{Q}$ is bivalent in $C_i\gamma\alpha_j\gamma''\beta_i$, it is bivalent in $C_i\gamma\alpha_j\gamma''d\beta_i$ as well.
	Property (\ref{blockswap}) (with $\lambda = \gamma\alpha_j\gamma''d$) implies that, if the block swap $\beta_i$ changes the value of some object in $\mathcal{Y}_i$ when applied in $C_i\gamma\alpha_j\gamma''d$, then $\mathcal{Q}$ is univalent in $C_i\gamma\alpha_j\gamma''d\beta_i$.
	This is a contradiction, because $\mathcal{Q}$ is bivalent in $C_i\gamma\alpha_j\gamma''d\beta_i$, $\textit{value}(B, C_i\gamma\alpha_j\gamma''d) = \bar{v}$, and $\textit{value}(B, C_i\gamma\alpha_j\gamma''d\beta_i) = v$.
	Hence, $B \not\in \mathcal{Y}_i$.
	
	\smallskip
	
	We have shown that $B \not\in \mathcal{X}_i \cup \mathcal{Y}_i$.
	Recall that $C_{i+1} = C_i\gamma\alpha_j$.
	Let $\mathcal{X}_{i+1} = \mathcal{X}_i$, $\mathcal{Y}_{i+1} = \mathcal{Y}_i \cup \{B\}$, $S_{i+1} = S_i \cup \{p_{i}\}$, and $\beta_{i+1} = d\beta_i$.
	Since $S_i$ takes no steps during $\gamma\alpha_j$, we know that $S_i$ still covers $\mathcal{Y}_i$ in $C_i\gamma\alpha_j = C_{i+1}$.
	Since $\alpha_j \widesim{p_i} \delta_j$, we know that $p_i$ is poised to apply $d$ in $C_i\gamma\alpha_j = C_{i+1}$.	
	Since $d$ is a \emph{Swap}$(B, \bar{v})$ operation, process $p_{i}$ covers $B$ in $C_{i+1}$.
	Hence, $S_{i+1} = S_i \cup \{p_{i}\}$ covers $\mathcal{Y}_{i+1} = \mathcal{Y}_i \cup \{B\}$ in $C_{i+1}$, which proves property (\ref{propCover}) for $i+1$.
	
	\smallskip

	If $\lambda'$ is a $(\mathcal{Q} \cup \mathcal{P}_{i+1})$-only execution from $C_{i+1} = C_i\gamma\alpha_j$, then $\gamma\alpha_j\lambda'$ is a $(\mathcal{Q} \cup \mathcal{P}_{i})$-only execution from $C_i$.
	By property~(\ref{nochange}) (with $\lambda = \gamma\alpha_j\lambda'$), if the value of some object in $\mathcal{X}_i = \mathcal{X}_{i+1}$ changes during $\gamma\alpha_j\lambda'$, then $\mathcal{P}$ is univalent in $C_i\gamma\alpha_j\lambda' = C_{i+1}\lambda'$.
	This proves property (\ref{nochange}) for $i+1$.
	
	\smallskip

	Notice that $\gamma\alpha_j\lambda'd$ is a $(\mathcal{Q} \cup \mathcal{P}_i)$-only execution from $C_i$.
	By property~(\ref{blockswap}) (with $\lambda = \gamma\alpha_j\lambda'd$), if $\beta_i$ changes the value of some object in $\mathcal{Y}_i$ when it is applied in $C_i\gamma\alpha_j\lambda'd$, then $\mathcal{Q}$ is univalent in $C_i\gamma\alpha_j\lambda'd\beta_i$.
	Lemma~\ref{lem:solotrick}\ref{lem:solotrick:cover} implies that, if $\textit{value}(B, C_i\gamma\alpha_j\lambda') = \textit{value}(B, C_i\gamma\delta_j) = v$, then $\mathcal{Q}$ is univalent in $C_i\gamma\alpha_j\lambda'd$.
	Recall that $d$ is a \emph{Swap}$(B, \bar{v})$ operation.
	Hence, if $d$ changes the value of $B$ when it is applied in $C_i\gamma\alpha_j\lambda'$, then $\mathcal{Q}$ is univalent in $C_i\gamma\alpha_j\lambda'd$.
	Therefore, if $d\beta_i = \beta_{i+1}$ changes the value of some object in $\mathcal{Y}_i \cup \{B\} = \mathcal{Y}_{i+1}$ when it is applied in $C_i\gamma\alpha_j\lambda' = C_{i+1}\lambda'$, then $\mathcal{Q}$ is univalent in $C_{i+1}\lambda'\beta_{i+1}$.
	This proves property~(\ref{blockswap}) for $i+1$.
	
\smallskip

In either case, the lemma holds for $i + 1$.
Therefore, the lemma holds for all $i \in \{0, \ldots, n-2\}$ by induction.
\end{proof}

Taking $i = n-2$ in Lemma~\ref{lem:maintech} gives us two disjoint sets of objects $\mathcal{X}_{n-2}$ and $\mathcal{Y}_{n-2}$ with $|\mathcal{X}_{n-2} \cup \mathcal{Y}_{n-2}| = n-2$.
This gives us the following theorem.

\begin{theorem}\label{thm:bswaplb}
	For all $n \geq 2$, any $n$-process, obstruction-free binary consensus algorithm from readable binary swap objects uses at least $n-2$ objects.
\end{theorem}

Ellen, Gelashvili, and Zhu \citep{egz-18} proved that any nondeterministic solo-terminating algorithm that uses a set of readable objects can be transformed into an obstruction-free algorithm that solves the same task and uses the same set of readable objects.
Hence, our lower bound also holds for nondeterministic solo-terminating (and hence, randomized wait-free) consensus algorithms.
Furthermore, a single readable binary swap object can be used to simulate any binary historyless object \citep{efr-07}.
Combined with Theorem~\ref{thm:bswaplb}, this gives us the following.

\begin{corollary}\label{cor:binhist}
	For all $n \geq 2$, any $n$-process, nondeterministic solo-terminating, binary consensus algorithm from binary historyless objects uses at least $n-2$ objects.
\end{corollary}

\subsection{Readable bounded swap objects}\label{sec:bounded}


Consider some $n$-process obstruction-free binary consensus algorithm from readable swap objects with domain size $b$.
We will now modify our technique from the proof of Lemma~\ref{lem:maintech} to show that this algorithm uses at least $\frac{n-2}{3b+1}$ objects.
Without loss of generality, suppose that the domain of every readable swap object is $\{0, \ldots, b-1\}$.
Let $\mathcal{A}$ denote the set of all readable swap objects used by the algorithm.
Let $C_0$ be an initial configuration of the algorithm in which $q_0$ has input $0$ and $q_1$ has input $1$.
As before, the pair of processes $\mathcal{Q}$ is bivalent in $C_0$.

Instead of constructing two sets of objects $\mathcal{X}_i$ and $\mathcal{Y}_i$, we construct two functions $f_i$ and $g_i$ that map objects to subsets of $\{0, \ldots, b-1\}$ along with a configuration $C_i$ in which $\mathcal{Q}$ is bivalent.
The functions $f_i$ and $g_i$ denote sets of forbidden values for each object in certain executions from $C_i$.
More specifically, for any $(\mathcal{Q} \cup \mathcal{P}_i)$-only execution $\lambda$ from $C_i$, if the value of some object $B$ is in $f_i(B)$ at any point during $\lambda$, then $\mathcal{Q}$ is univalent in $C_i\lambda$.
If a process $p \in \mathcal{P}_i$ is poised to apply a \emph{Swap}$(B, x)$ operation in $C_i\lambda$ for some object $B$ and some $x \in g_i(B)$, then $\mathcal{Q}$ is univalent in $C_i\lambda$.
We also obtain a set of processes $S_i \subseteq \mathcal{P} - \mathcal{P}_i$ that covers a set of $|S_i|$ objects in $C_i$.
If a process in $S_i$ is poised to apply a \emph{Swap}$(B, x)$ operation in $C_i$, then $x \not\in f_i(B) \cup g_i(B)$.
Unlike in Lemma~\ref{lem:maintech}, it is not necessarily true that $S_i \subseteq S_{i+1}$.
In particular, we might remove a single process $p$ from $S_i$ to obtain $S_{i+1}$.
However, whenever we do this, we also add a new value to $f_i(B)$ to obtain $f_{i+1}(B)$, where $B$ is the object covered by $p$ in $C_i$.
This allows us to show that $\sum_{B \in \mathcal{A}}\bigl(2\cdot |f_i(B)| + |g_i(B)|\bigr) + |S_i| \geq i$.

\medskip

To concretely illustrate how our technique in this section differs from the proof of Lemma~\ref{lem:maintech}, we explain how to obtain $C_i$, $S_i$, and the sets of forbidden values $f_i$ and $g_i$ for $i = 1$.
Let $\delta$ be $p_{0}$'s solo-terminating execution from $C_0$.
Suppose that $\delta$ consists of $r$ steps and, for all $s \in \{0, \ldots, r\}$, let $\delta_s$ be the prefix of $\delta$ that consists of the first $s$ steps by $p_0$.
Let $j$ be the value that satisfies the conditions of Lemma~\ref{lem:solotrick} (with $C = C' = C_0$ and $i = 0$).
Lemma~\ref{lem:solotrick}(\ref{lem:solotrick:prefix}) implies that there is an execution $\alpha_j$ from $C_0$ such that $\mathcal{Q}$ is bivalent in $C_0\alpha_j$ and $\alpha_j \widesim{p_0} \delta_j$.
Define $C_1 = C_0\alpha_j$.
Then $\mathcal{Q}$ is bivalent in $C_1$.

Let $B^\star$ be the base object that $p_0$ is poised to access in $C_0\delta_j$ and let $d$ be the operation that $p_0$ is poised to apply.
Let $v^\star = \textit{value}(B^\star, C_0\delta_j)$.
If the application of $d$ in $C_0\delta_j$ does not change the value of $B^\star$, i.e. $v^\star = \textit{value}(B^\star, C_0\delta_j d)$, then define $f_1(B^\star) = \{v^\star\}$, $f_1(B) = \emptyset$ for all $B \in \mathcal{A} - \{B^\star\}$, and $g_1(B) = \emptyset$ for all $B \in \mathcal{A}$.
Define $S_1 = \emptyset$.
Unlike in the proof of Lemma~\ref{lem:maintech}, it is not possible to show that, for any $(\mathcal{Q} \cup \mathcal{P}_1)$-only execution $\lambda'$ from $C_1$, $\mathcal{Q}$ is univalent in $C_1\lambda'$ if the value of $B^\star$ changes during $\lambda'$.
However, Lemma~\ref{lem:solotrick}\ref{lem:solotrick:nochange} implies that, if the object $B^\star$ has the value $v$ at some point during $\lambda'$, then $\mathcal{Q}$ is univalent in $C_0\alpha_j\lambda' = C_1\lambda'$.

Otherwise, the application of $d$ in $C_0\delta_j$ changes the value of $B^\star$, i.e. $v^\star \neq \textit{value}(B^\star, C_0\delta_j d)$.
Define $f_1(B) = \emptyset$ for all $B \in \mathcal{A}$, $g_1(B^\star) = \{v^\star\}$, and $g_1(B) = \emptyset$ for all $B \in \mathcal{A} - \{B^\star\}$.
Define $S_1 = \{p_0\}$.
Since $d$ changes the value of $B^\star$ when applied in $C_0\delta_j$, we know that $d$ is a \emph{Swap} operation.
Hence, $S_1 = \{p_1\}$ covers a set of $|S_1|$ objects $\{B^\star\}$ in $C_0\alpha_j = C_1$.
Unlike in the proof of Lemma~\ref{lem:maintech}, it is not possible to show that, for any $(\mathcal{Q} \cup \mathcal{P}_1)$-only execution $\lambda'$ from $C_1$, $\mathcal{Q}$ is univalent in $C_1\lambda'd$ if $d$ changes the value of $B^\star$ when applied in $C_1\lambda'$.
However, Lemma~\ref{lem:solotrick}\ref{lem:solotrick:cover} implies that, if the value of $B^\star$ is $v^\star$ in $C_1\lambda'$, then $\mathcal{Q}$ is univalent in $C_1\lambda'd$.
Now suppose there exists a $(\mathcal{Q} \cup \mathcal{P}_1)$-only execution $\lambda$ from $C_1$ such that there is a process $p_i \in \mathcal{P}_1$ poised to apply a \emph{Swap}$(B^\star, v^\star)$ operation $d^\star$ in $C_1\lambda$.
If $\mathcal{Q}$ is bivalent in $C_1\lambda$, then we can apply Lemma~\ref{lem:extend} (with $S = \{p_0\}$) to obtain a $\mathcal{Q}$-only execution $\gamma$ from $C_1\lambda$ such that $\mathcal{Q}$ is bivalent in $C_1\lambda\gamma d$.
Then $\mathcal{Q}$ is bivalent in $C_1\lambda\gamma d^\star d$ as well, which is a contradiction.
Hence, $\mathcal{Q}$ is univalent in $C_1\lambda$.

\medskip

When $i \geq 1$, we insert a block swap $\beta_i$ by $S_i$ before the solo execution $\delta$ by $p_i$ in the induction step, as depicted in Figure~\ref{fig:rsobounded}.
The reason for this is that, unlike in the proof of Lemma~\ref{lem:maintech}, there is no way to construct a $\mathcal{Q}$-only execution $\gamma$ from $C_i$ such that $\mathcal{Q}$ is bivalent in both $C_i\gamma$ and $C_i\gamma\beta_i$.
Instead, we use the fact that the block swap $\beta_i$ can only swap non-forbidden values into the covered base objects.
Therefore, if $p_i$ applies a \emph{Read} or a \emph{Swap} during $\delta$ to one of the base objects covered by $S_i$ in $C_i$ and obtains a forbidden value $v$ as a response, then $p_i$ must have changed the value of that base object to $v$ in some previous step of $\delta$.
Once again, we use Lemma~\ref{lem:solotrick} to obtain an execution $\alpha_j$ that is indistinguishable from the first $j$ steps $\delta_j$ of $\delta$ to process $p_i$.
We will show that it is not possible for $p_i$ to apply any \emph{Swap}$(B, v)$ operations in the first $j+1$ steps of $\delta$, for any base object $B$ and any forbidden value $v \in f_i(B) \cup g_i(B)$.
Then we show that, if the last step of $\delta_{j+1}$ by $p_i$ accesses the base object $B^\star$ and does not change its value, then we can add \emph{value}$(B^\star, C_i\delta_j)$ to $f_i(B^\star)$ to obtain $f_{i+1}(B^\star)$.
If $B^\star$ is covered by some process in $S_i$ in $C_i$, then we remove that process from $S_i$ to obtain $S_{i+1}$.
On the other hand, if the last step of $\delta_{j+1}$ by $p_i$ changes the value of $B^\star$, then we add \emph{value}$(B^\star, C_i\delta_j)$ to $g_i(B^\star)$ to obtain $g_{i+1}(B^\star)$.
If $B^\star$ is covered by some process in $S_i$ in $C_i$, then we replace this process with $p_i$ to obtain $S_{i+1}$.
Otherwise, we add $p_i$ to $S_i$ to obtain $S_{i+1}$.
We formalize this argument in the proof of the following lemma.

\begin{lemma}\label{lem:tech}
For all $i \in \{0, \ldots, n-2\}$, there is a configuration $C_i$ reachable from $C_0$, a set of processes $S_i \subseteq \mathcal{P} - \mathcal{P}_i$, and a pair of functions $f_i, g_i$ that map base objects to subsets of $\{0, \ldots, b-1\}$ such that the following properties hold for every $(\mathcal{Q} \cup \mathcal{P}_i)$-only execution $\lambda$ from $C_i$:

\begin{enumerate}[label=(\alph*),ref=\alph*]
	\item $\mathcal{Q}$ is bivalent in $C_i$,\label{lem:qbiv}
	\item $S_i$ covers a set of $|S_i|$ base objects in $C_i$,\label{lem:ricover}
	\item for every process $p \in S_i$, if $p$ is poised to apply a \textit{Swap}$(B, x)$ operation in $C_i$, then $x \not\in f_i(B) \cup g_i(B)$,\label{lem:richange}
	\item $\sum_{B \in \mathcal{A}} \bigl(2\cdot |f_i(B)| + |g_i(B)|\bigr) + |S_i| \geq i$,\label{lem:size}

	\item if the value of some object $B$ is equal to some value in $f_i(B)$ in any configuration of $\lambda$, then $\mathcal{Q}$ is univalent in $C_i\lambda$, and\label{lem:fvals}
	\item if some process $p \in \mathcal{P}_i$ is poised to apply a \textit{Swap}$(B, x)$ operation in $C_i\lambda$ for some object $B$ and some $x \in g_i(B)$, then $\mathcal{Q}$ is univalent in $C_i\lambda$.\label{lem:gvals}
\end{enumerate}
\end{lemma}

\begin{proof}
We use induction on $i$.
Let $C_0$ be the bivalent initial configuration defined earlier in which $q_0$ has input $0$ and $q_1$ has input $1$.
This gives us property~(\ref{lem:qbiv}).
Let $S_0 = \emptyset$ and let $f_0(B) = g_0(B) = \emptyset$ for all $B \in \mathcal{A}$.
Properties~(\ref{lem:ricover}), (\ref{lem:richange}), (\ref{lem:fvals}), and (\ref{lem:gvals}) all hold vacuously since $S_0$ is empty and $f_0(B)$ and $g_0(B)$ are empty for all $B \in \mathcal{A}$.
Property~(\ref{lem:size}) follows from $\sum_{B \in \mathcal{A}} \bigl(2\cdot |f_0(B)| + |g_0(B)|\bigr) + |S_0| = 0$.

Now suppose that the lemma holds for some $i \in \{0, \ldots, n-3\}$.
Let $\delta$ be $p_i$'s solo-terminating execution from $C_i\beta_i$, where $\beta_i$ is a block swap by $S_i$.
Suppose that $\delta$ consists of $r$ steps by $p_i$, and for all $s \in \{0, \ldots, r\}$, let $\delta_s$ be the prefix of $\delta$ that consists of the first $s$ steps by $p_i$.
Let $j \in \{0, \ldots, r-1\}$ be the value that satisfies the conditions of Lemma~\ref{lem:solotrick} (with $C = C_i$ and $C' = C_i\beta_i$).
The following claim is important for our construction.

\begin{claim}\label{claim:offlimits}
	For all $B \in \mathcal{A}$ and all $x \in f_i(B) \cup g_i(B)$, process $p_i$ does not apply any \textit{Swap}$(B, x)$ operations in $\delta_{j+1}$.
\end{claim}

\begin{proof}[Proof of Claim~\ref{claim:offlimits}]
To obtain a contradiction, suppose that for some $B \in \mathcal{A}$, some $x \in f_i(B) \cup g_i(B)$, and some $0 \leq t \leq j$, $p_i$ is poised to apply a \emph{Swap}$(B, x)$ operation in $C_i\beta_i\delta_{t}$.
Since $0 \leq t \leq j$, Lemma~\ref{lem:solotrick}(\ref{lem:solotrick:prefix}) implies that there exists a $(\mathcal{Q} \cup \mathcal{P}_i)$-only execution $\alpha_{t}$ from $C_i$ such that $\mathcal{Q}$ is bivalent in $C_i\alpha_{t}$ and $\alpha_{t} \widesim{p_i} \delta_{t}$.
Then $p_i$ is poised to apply a \emph{Swap}$(B, x)$ operation in $C_i\alpha_{t}$.
If $x \in g_i(B)$, then property~(\ref{lem:gvals}) (with $\lambda = \alpha_{t}$) implies that $\mathcal{Q}$ is univalent in $C_i\alpha_{t}$.
This is a contradiction.
Therefore, $x \in f_i(B)$.

Let $d_t$ be the step that $p_i$ is poised to apply in $C_i\beta_i\delta_{t}$ and $C_i\alpha_{t}$.
By Lemma~\ref{lem:extend} (with $C = C_i\alpha_{t}$ and $S = \{p_i\}$), there is a $\mathcal{Q}$-only execution $\gamma$ from $C_i\alpha_{t}$ such that $\mathcal{Q}$ is bivalent in $C_i\alpha_{t}\gamma d_t$.
Since $d_t$ is a \emph{Swap}$(B, x)$ operation, \emph{value}$(B, C_i\alpha_{t}\gamma d_t) = x \in f_i(B)$.
Then by property~(\ref{lem:fvals}) (with $\lambda = \alpha_{t}\gamma d_t$), $\mathcal{Q}$ is univalent in $C_i\alpha_t\gamma d_t$.
This is a contradiction.
\end{proof}

\begin{figure}[h]
\centering
\begin{tikzpicture}[shorten >=1pt,node distance=2.9cm,on grid,auto] 
	\tikzset{every state/.style={minimum size=1.1cm}}

	\node[state,accepting]	(Ci)	{$C_i$};
	\node[state,accepting]	(alpha) [right=4.5cm of Ci]	{$C_{i+1}$};
	\node[state]	(beta)	[above=2.2cm of Ci]	{};
	\node[state]	(delta)	[above=2.2cm of alpha]	{};
	
	\node	(indist)	[below=1.1cm of delta]	{$\delta_j \widesim{p_i} \alpha_j$};

    \path[->,
	line join=round, decoration={
	    zigzag,
	    segment length=4,
	    amplitude=.9,post=lineto,
	    post length=2pt}]	    
	    
	(Ci)	edge[decorate] node[below] {$(\mathcal{Q} \cup \mathcal{P}_i)$-only}  node[above] {$\alpha_j$} (alpha)
	(Ci)	edge node[right] {\small $S_i$-only} node[left] {$\beta_i$} (beta)
	(beta)	edge[decorate] node[below] {$p_i$-only} node[above] {$\delta_j$} (delta);
\end{tikzpicture}
\caption{The construction of $C_{i+1}$ from $C_i$ in the proof of Lemma~\ref{lem:tech}. Nodes with double outlines denote configurations in which $\mathcal{Q}$ is bivalent.}\label{fig:rsobounded}
\end{figure}
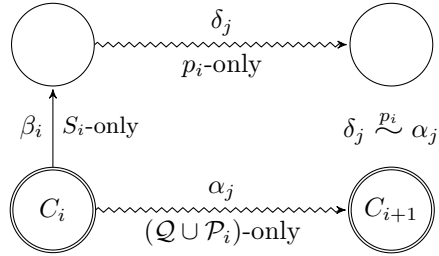

We now proceed with the inductive step.
Let $d$ be the operation that $p_i$ is poised to apply to the object $B^\star$ in $C_i\beta_i\delta_j$.
Let $v^\star = \textit{value}(B^\star, C_i\beta_i\delta_j)$.
Define $C_{i+1} = C_i\alpha_{j}$.
This is illustrated in Figure~\ref{fig:rsobounded}.
Since $\mathcal{Q}$ is bivalent in $C_i\alpha_{j}$, we obtain property~(\ref{lem:qbiv}) for $i+1$.
There are two cases.

\smallskip

\textbf{Case 1:} $\textit{value}(B^\star, C_i\beta_i\delta_j d) = v^\star$.
	Since $d$ does not change the value of $B$,  $d$ is either a \emph{Read}$(B^\star)$ or a \emph{Swap}$(B^\star, v^\star)$ operation.
	In this case, we add $v^\star$ to $f_i(B^\star)$ to obtain $f_{i+1}(B^\star)$.
	More formally, define $f_{i+1}(B) = f_i(B)$ for all $B \in \mathcal{A} - \{B^\star\}$, $g_{i+1}(B) = g_i(B)$ for all $B \in \mathcal{A}$, and $f_{i+1}(B^\star) = f_i(B^\star) \cup \{v^\star\}$.
	If there is a process $p \in S_i$ that covers $B^\star$ in $C_i$ and $p$ is poised to apply \emph{Swap}$(B^\star, v^\star)$ in $C_i$, then define $S_{i+1} = S_i - \{p\}$.
	Otherwise, define $S_{i+1} = S_i$.
	Since $S_i$ covers $|S_i|$ objects in $C_i$ by property~(\ref{lem:ricover}) of the inductive hypothesis, no process in $S_i$ takes steps during $\alpha_{j}$, and $S_{i+1} \subseteq S_i$, this gives us property~(\ref{lem:ricover}) for $i+1$.
	
	\smallskip
	
	Notice that if some process $p \in S_{i+1}$ is poised to apply a \emph{Swap}$(B^\star, x)$ operation in $C_{i+1}$, then $x \neq v^\star$ by definition of $S_{i+1}$.
	Furthermore, since none of the processes in $S_{i+1}$ took steps in $\alpha_j$, they are poised to apply the same operations in $C_{i+1}$ as in $C_i$.	
	Hence, property~(\ref{lem:richange}) of the inductive hypothesis gives us property~(\ref{lem:richange}) for $i+1$.
	
	\smallskip
	
	We now show that $v^\star \not\in f_i(B^\star)$.
	To obtain a contradiction, suppose that $v^\star \in f_i(B^\star)$.
	By Claim~\ref{claim:offlimits}, process $p_i$ cannot change the value of $B^\star$ to a value in $f_i(B^\star)$ during $\delta_j$.
	Hence, \emph{value}$(B^\star, C_i\beta_i) = v^\star$.
	If $B^\star$ is covered by $S_i$ in $C_i$, then the process $p \in S_i$ that swaps $B^\star$ during $\beta_i$ applies \emph{Swap}$(B^\star, v^\star)$ in this block swap.
	Since $v^\star \in f_i(B^\star)$, this contradicts property~(\ref{lem:richange}) of the inductive hypothesis.
	Thus, $B^\star$ is not covered by $S_i$ in $C_i$.
	Then \emph{value}$(B^\star, C_i) = v^\star$.
	Property~(\ref{lem:fvals}) of the inductive hypothesis (where $\lambda$ is the empty execution) implies that $\mathcal{Q}$ is univalent in $C_i$.
	This contradicts property~(\ref{lem:qbiv}) of the inductive hypothesis.
	Hence, $v^\star \not\in f_i(B^\star)$.
	This implies that $|f_{i+1}(B^\star)| = |f_i(B^\star)| + 1$.
	Since $|S_{i+1}| \geq |S_i| - 1$, this gives us property~(\ref{lem:size}) for $i+1$.
	
	\smallskip
	
	Consider a $(\mathcal{Q} \cup \mathcal{P}_{i+1})$-only execution $\lambda''$ from $C_{i+1} = C_i\alpha_{j}$.
	Then $\alpha_j\lambda''$ is a $(\mathcal{Q} \cup \mathcal{P}_i)$-only execution from $C_i$.
	Hence, by property~(\ref{lem:fvals}) (with $\lambda = \alpha_j\lambda''$), if the value of some object $B$ is equal to some value in $f_i(B)$ in any configuration of $\alpha_j\lambda''$, then $\mathcal{Q}$ is univalent in $C_{i}\alpha_j\lambda'' = C_{i+1}\lambda''$.
	If the value of $B$ is equal to $v^\star$ at some point during $\lambda''$, then Lemma~\ref{lem:solotrick}\ref{lem:solotrick:nochange} (with $C = C_i$, $C' = C_i\beta_i$, and $\lambda' = \lambda''$) implies that $\mathcal{Q}$ is univalent in $C_{i+1}\lambda''$.
	This completes the proof of property~(\ref{lem:fvals}) for $i+1$.
	
	Property~(\ref{lem:gvals}) (with $\lambda = \alpha_j\lambda''$) says that, if some process $p \in \mathcal{P}_i$ is poised to apply a \emph{Swap}$(B, x)$ operation in $C_i\alpha_j\lambda''$ for some object $B$ and some $x \in g_i(B) = g_{i+1}(B)$, then $\mathcal{Q}$ is univalent in $C_i\alpha_j\lambda'' = C_{i+1}\lambda''$.
	This gives us property~(\ref{lem:gvals}) for $i+1$.
	

\smallskip

\textbf{Case 2:} \emph{value}$(B^\star, C_i\beta_i\delta_j d) \neq \textit{value}(B^\star, C_i\beta_i\delta_j) = v^\star$.
	Then $d$ is a \emph{Swap}$(B^\star, v')$ operation, for some $v' \in \{0, \ldots, b-1\} - \{v^\star\}$.
	In this case, we add $v^\star$ to $g_i(B^\star)$ to obtain $g_{i+1}(B^\star)$.
	More formally, define $f_{i+1}(B) = f_i(B)$ for all $B \in \mathcal{A}$, $g_{i+1}(B) = g_i(B)$ for all $B \in \mathcal{A} - \{B^\star\}$, and $g_{i+1}(B^\star) = g_i(B^\star) \cup \{v^\star\}$.
	
	If some process $p \in S_i$ is poised to access $B^\star$ in $C_i$, then define $S_{i+1} = (S_i - \{p\}) \cup \{p_i\}$.
	In this case, $S_{i+1}$ covers the same set of objects in $C_{i+1}$ as $S_i$ covers in $C_i$ and $|S_{i+1}| = |S_i|$, which gives us property~(\ref{lem:ricover}) for $i+1$.
	Otherwise, define $S_{i+1} = S_i \cup \{p_i\}$.
	In this case, $S_{i+1}$ covers the same objects in $C_{i+1}$ as $S_i$ covers in $C_i$ in addition to the object $B^\star$.
	This also gives us property~(\ref{lem:ricover}) for $i+1$.
	
	\smallskip
	
	By Claim~\ref{claim:offlimits}, we have $v' \not\in \bigl(f_i(B^\star) \cup g_i(B^\star)\bigr)$.
	Furthermore, we know that $v' \neq v^\star$.
	Hence, $v' \not\in \bigl(f_{i+1}(B^\star) \cup g_{i+1}(B^\star)\bigr)$.
	Notice that no process in $S_i$ takes any steps in $\alpha_j$, $f_{i+1}(B) = f_i(B)$, and $g_{i+1}(B) = g_i(B)$ for all $B \in \mathcal{A} - \{B^\star\}$.
	Hence, property~(\ref{lem:richange}) of the inductive hypothesis implies that, for every process $p \in S_i$, if $p$ is poised to apply \emph{Swap}$(B, x)$ in $C_{i+1}$ for some $B \in \mathcal{A} - \{B^\star\}$, then $x \not\in \bigl(f_{i+1}(B) \cup g_{i+1}(B)\bigr)$.
	Since $p_i$ is the only process in $S_{i+1}$ poised to swap $B^\star$ in $C_{i+1}$, this gives us property~(\ref{lem:richange}) for $i+1$.
	
	\smallskip
	
	Note that it is not guaranteed that $v^\star \not\in g_i(B^\star)$.
	In other words, we might not add a new forbidden value in this case.
	However, we now show that either $|S_{i+1}| = |S_i| + 1$ or $|g_{i+1}(B^\star)| = |g_i(B^\star)| + 1$, which allows us to obtain property~(\ref{lem:size}) for $i+1$.
	
	If $B^\star$ is not covered by $S_i$ in $C_i$, then $S_{i+1} = S_i \cup \{p_i\}$.
	Hence, $|S_{i+1}| = |S_i| + 1$.
	Furthermore, $|f_{i+1}(B)| = |f_i(B)|$ for all $B \in \mathcal{A}$, $|g_{i+1}(B)| = |g_i(B)|$ for all $B \in \mathcal{A} - \{B^\star\}$, and $|g_{i+1}(B^\star)| \geq |g_i(B^\star )|$.
	This gives us property~(\ref{lem:size}) for $i+1$.
	
	Otherwise, $B^\star$ is covered by $S_i$ in $C_i$.
	By property~(\ref{lem:richange}) of the inductive hypothesis, we have \emph{value}$(B, C_i\beta_i) \not\in \bigl(f_i(B) \cup g_i(B)\bigr)$ for all objects $B$ covered by $S_i$ in $C_i$.
	Furthermore, by Claim~\ref{claim:offlimits}, process $p_i$ does not change the value of any object $B \in \mathcal{A}$ to any value in $f_i(B) \cup g_i(B)$ during $\delta_{j+1}$.
	Then \emph{value}$(B, C_i\beta_i\delta_j) \not\in f_i(B) \cup g_i(B)$ for all objects $B$ covered by $S_i$ in $C_i$.
	In particular, since $B^\star$ is covered by $S_i$ in $C_i$, this implies $v^\star \not\in f_i(B^\star) \cup g_i(B^\star)$.
	Since $g_{i+1}(B^\star) = g_{i}(B) \cup \{v^\star\}$, this implies $|g_{i+1}(B^\star)| = |g_i(B^\star)| + 1$.
	Furthermore, since $S_{i+1} = (S_i - \{p\}) \cup \{p_i\}$, where $p$ is the process in $S_i$ that covers $B^\star$ in $C_i$, we have $|S_{i+1}| = |S_i|$.
	Finally, $|f_{i+1}(B)| = |f_i(B)|$ for all objects $B$, and $|g_{i+1}(B)| = |g_{i}(B)|$ for all $B \in \mathcal{A} - \{B^\star\}$.
	This gives us property~(\ref{lem:size}) for $i+1$.
	
	\smallskip

	Consider a $(\mathcal{Q} \cup \mathcal{P}_{i+1})$-only execution $\lambda''$ from $C_{i+1}$.
	Then $\alpha_{j}\lambda''$ is a $(\mathcal{Q} \cup \mathcal{P}_i)$-only execution from $C_i$.
	Hence, by property~(\ref{lem:fvals}) (with $\lambda = \alpha_j\lambda''$), if the value of some object $B$ is equal to some value in $f_i(B) = f_{i+1}(B)$ in any configuration of $\alpha_{j}\lambda''$, then $\mathcal{Q}$ is univalent in $C_i\alpha_j\lambda'' = C_{i+1}\lambda''$.
	This gives us property~(\ref{lem:fvals}) for $i+1$.
	
	By property~(\ref{lem:gvals}) (with $\lambda = \alpha_j\lambda''$), if some process $p \in \mathcal{P}_{i+1} \subsetneq \mathcal{P}_i$ is poised to apply a \emph{Swap}$(B, x)$ operation in $C_{i}\alpha_j\lambda''$ for some object $B$ and some $x \in g_i(B)$, then $\mathcal{Q}$ is univalent in $C_{i}\alpha_j\lambda''$.
	We now prove that if some process $p \in \mathcal{P}_{i+1}$ is poised to apply a \emph{Swap}$(B^\star, v^\star)$ operation in $C_{i+1}\lambda''$, then $\mathcal{Q}$ is univalent in $C_{i+1}\lambda''$.
	To obtain a contradiction, suppose that $\mathcal{Q}$ is bivalent in $C_{i+1}\lambda''$ and some process $p \in \mathcal{P}_{i+1}$ is poised to apply a \emph{Swap}$(B^\star, v^\star)$ operation $t$ in $C_{i+1}\lambda''$.
	By Lemma~\ref{lem:extend} (with $C = C_{i+1}\lambda''$ and $S = \{p_i\}$), there is a $\mathcal{Q}$-only execution $\gamma$ from $C_{i+1}\lambda''$ such that $\mathcal{Q}$ is bivalent in $C_{i+1}\lambda''\gamma d$.
	Since $t$ and $d$ are both applied to $B^\star$, all of the objects have the same values in $C_{i+1}\lambda''\gamma d$ and $C_{i+1}\lambda''\gamma t d$.
	Then $\mathcal{Q}$ is bivalent in $C_{i+1}\lambda''\gamma t d$ as well.
	However, Lemma~\ref{lem:solotrick}\ref{lem:solotrick:cover} (with $\lambda' = \lambda''\gamma t$) implies that $\mathcal{Q}$ is univalent in $C_{i+1}\lambda''\gamma td$, which is a contradiction.
	Hence, if some process $p \in \mathcal{P}_{i+1}$ is poised to apply a \emph{Swap}$(B^\star, v^\star)$ operation in $C_{i+1}\lambda''$, then $\mathcal{Q}$ is univalent in $C_{i+1}\lambda''$.
	This completes the proof of property~(\ref{lem:gvals}) for $i+1$.
\end{proof}

Applying Lemma~\ref{lem:tech} with $i = n-2$ gives us $f_{n-2}, g_{n-2}$, and $S_{n-2}$ with $\sum_{B \in \mathcal{A}} \bigl(2\cdot |f_{n-2}(B)| + |g_{n-2}(B)|\bigr) + |S_{n-2}| \geq n-2$ by property~(\ref{lem:size}).
Since $f_{n-2}(B)$ and $g_{n-2}(B)$ are subsets of $\{0, \ldots b-1\}$, $\sum_{B \in \mathcal{A}} \bigl(2\cdot |f_{n-2}(B)| + |g_{n-2}(B)|\bigr) \leq 3\cdot b\cdot |\mathcal{A}|$.
Since $S_{n-2}$ covers a set of $|S_{n-2}|$ objects in $C_{n-2}$ by property~(\ref{lem:ricover}), $|S_{n-2}| \leq |\mathcal{A}|$.
Thus, $3\cdot b\cdot |\mathcal{A}| + |\mathcal{A}| \geq n-2$.
This gives us the following theorem.

\begin{theorem}\label{thm:boundedlb}
For all $n, b \geq 2$, any $n$-process, obstruction-free binary consensus algorithm from readable swap objects with domain size $b$ uses at least $\frac{n-2}{3b + 1}$ objects.
\end{theorem}

Once again, a result by Ellen, Gelashvili, and Zhu \citep{egz-18} implies that our lower bound also holds for nondeterministic solo-terminating (and hence, randomized wait-free) consensus algorithms.
This, along with the fact that a readable swap objects can be used to simulate any historyless object with the same domain \citep{efr-07}, gives us the following.

\begin{corollary}
	For all $n \geq 2$, any $n$-process, nondeterministic solo-terminating, binary consensus algorithm from historyless objects with domain size $b$ uses at least $\frac{n-2}{3b+1}$ objects.
\end{corollary}

\section{Conclusion}\label{sec:conclusion}

In this paper, we showed that $n$-process obstruction-free $k$-set agreement can be solved with $n-k$ swap objects with unbounded domains.
We also proved a lower bound of $\lceil \frac{n}{k}\rceil - 1$ swap objects for solving nondeterministic solo-terminating $k$-set agreement, which exactly matches our algorithm for $k = 1$.
Since a swap operation can simulate any nontrivial operation on a historyless object, this lower bound implies that, if there is a nondeterministic solo-terminating $k$-set agreement algorithm from historyless objects that uses fewer than $\lceil\frac{n}{k}\rceil - 1$ objects, then at least one of the historyless objects must support a trivial operation.
Closing the gap between these upper and lower bounds remains an open problem.
We conjecture that at least $n-k$ swap objects are necessary.

The optimal space complexity of solving $n$-process obstruction-free $k$-set agreement with registers is unknown when $k > 1$.
The best known upper bound, due to Bouzid, Raynal, and Sutra \citep{brs-18}, is $n-k+1$, while the best known lower bound, due to Ellen, Gelashvili, and Zhu \citep{egz-18}, is $\lceil \frac{n}{k}\rceil$.
Furthermore, there is no known non-constant lower bound on the space complexity of solving obstruction-free $k$-set agreement using readable swap objects.
One possible line of future work is to consider $k$-set agreement algorithms using readable swap objects with bounded domain sizes.

We also proved that any obstruction-free binary consensus algorithm from readable binary swap objects requires at least $n-2$ objects.
We modified the technique from this proof in order to show that $\frac{n-2}{3b+1}$ readable swap objects with domain size $b$ are needed to solve obstruction-free binary consensus.
When $b$ is a constant, this lower bound is $\Omega(n)$, which asymptotically matches the best known algorithms.
If $b$ is $o(\sqrt{n})$, then our lower bound is asymptotically larger than the $\Omega(\sqrt{n})$ lower bound of Ellen, Herlihy, and Shavit \cite{ehs-98}.
Obtaining an $\omega (\sqrt{n})$ lower bound on the space complexity of solving obstruction-free consensus using readable swap objects with unbounded domain is a longstanding open problem that will likely require different techniques.

Bowman \citep{b-11} presented an obstruction-free $n$-process binary consensus algorithm from $2n-1$ different binary registers.
Although our lower bound of $n-2$ asymptotically matches Bowman's algorithm, it remains open to prove exactly matching upper and lower bounds for binary historyless base objects.
It would be interesting to investigate whether there is an obstruction-free $n$-process binary consensus algorithm using fewer than $2n-1$ instances of binary swap objects rather than binary registers.

Another open question is determining the optimal space complexity of solving obstruction-free $n$-valued consensus using binary historyless objects.
The best known upper bound, due to Ellen, Gelashvili, Shavit, and Zhu \citep{egsz-20}, is $O(n\log n)$.
However, the best known lower bound for binary registers is $n$, due to Ellen, Gelashvili, and Zhu \citep{egz-18}, and the best known lower bound for arbitrary binary historyless objects is $n-2$, due to Corollary~\ref{cor:binhist}.

\ifarxiv
\bibliographystyle{unsrt}
\fi
\bibliography{refs}

\end{document}